\documentclass[11pt,a4paper,reqno]{amsart}%
\usepackage{amsbsy,amsopn,amscd,amstext,amsgen,amsthm,amsmath,amsfonts,amssymb,amsxtra,appendix,bookmark,dsfont,latexsym,color,epsfig,xspace,mathrsfs,dsfont,srcltx}

\usepackage{hyperref}
\usepackage{delarray,a4,color}
\usepackage{euscript}
\usepackage[latin1]{inputenc}

\makeatletter
\def\@tocline#1#2#3#4#5#6#7{\relax
  \ifnum #1>\c@tocdepth 
  \else
    \par \addpenalty\@secpenalty\addvspace{#2}%
    \begingroup \hyphenpenalty\@M
    \@ifempty{#4}{%
      \@tempdima\csname r@tocindent\number#1\endcsname\relax
    }{%
      \@tempdima#4\relax
    }%
    \parindent\z@ \leftskip#3\relax \advance\leftskip\@tempdima\relax
    \rightskip\@pnumwidth plus4em \parfillskip-\@pnumwidth
    #5\leavevmode\hskip-\@tempdima
      \ifcase #1
       \or\or \hskip 1em \or \hskip 2em \else \hskip 3em \fi%
      #6\nobreak\relax
      \dotfill
      \hbox to\@pnumwidth{\@tocpagenum{#7}}
    \par
    \nobreak
    \endgroup
  \fi}
\makeatother

\theoremstyle{plain}
\newtheorem{theorem}{Theorem}[section]
\newtheorem{lemma}[theorem]{Lemma}

\newtheorem{proposition}[theorem]{Proposition}

\newtheorem{assumption}[theorem]{Assumption}

\theoremstyle{remark}
\newtheorem{remark}[theorem]{Remark}

\numberwithin{equation}{section}

\DeclareMathOperator{\spec}{spec}

\DeclareMathOperator{\supp}{supp}

\DeclareMathOperator{\tr}{Tr}

\def\geqslant{\ge}
\def\leqslant{\le}
\def\bq{\begin{eqnarray}}
\def\eq{\end{eqnarray}}
\def\bqq{\begin{eqnarray*}}
\def\eqq{\end{eqnarray*}}

\def\eps{\varepsilon}

\newcommand{\norm}[1]{\left\lVert #1 \right\rVert}

\newcommand\1{{\ensuremath {\mathds 1} }}

\newcommand{\gammaP}{\gamma_{\Psi}}

\renewcommand{\epsilon}{\varepsilon}

\def\cF {\mathcal{F}}

\def\R {\mathbb{R}}

\def\C {\mathbb{C}}
\def\N {\mathcal{N}}

\def\E {\mathcal{E}}
\def\OO {\mathcal{O}}

\def\R {\mathbb{R}}
\def\C {\mathbb{C}}
\def\N {\mathcal{N}}

\def\gS{\mathfrak{S}}

\def\E {\mathcal{E}}

\def\gH{\mathfrak{H}}
\def\gHp{\mathfrak{H} _{\perp}}

\def\bH{\mathbb{H}}
\def\dd{\partial}
\newcommand{\Real}{\mathrm{Re}}
\newcommand{\dGamma}{{\ensuremath{\rm d}\Gamma}}
\newcommand\ket[1]{{\ensuremath{|#1\rangle}}}

\renewcommand{\leq}{\leqslant}
\renewcommand{\geq}{\geqslant}

\newcommand{\Psidloc}{\Psi_{\rm dloc}}
\newcommand{\Psiloc}{\Psi_{\rm loc}}
\newcommand{\gammai}{\gamma ^{(1)}}

\newcommand{\gammaidloc}{\gammai_{\rm dloc}}
\newcommand{\gammailoc}{\gammai_{\rm loc}}

\newcommand{\EH}{\E_{\rm H}}
\newcommand{\EHp}{\E_{{\rm H}+}}
\newcommand{\EHm}{\E_{{\rm H}-}}
\newcommand{\EHpm}{\E_{{\rm H}\pm}}
\newcommand{\eH}{e_{\rm H}}
\newcommand{\uH}{u_{\rm H}}
\newcommand{\uHm}{u_{{\rm H} -}}
\newcommand{\uHp}{u_{{\rm H} +}}
\newcommand{\uHpm}{u_{{\rm H} \pm}}
\newcommand{\uHmp}{u_{{\rm H} \mp}}
\newcommand{\muH}{\mu}

\newcommand{\HH}{H_{ \rm mf}}
\newcommand{\hH}{h_{ \rm mf}}

\newcommand{\EB}{\mathcal{E}_{\rm B}}
\newcommand{\eB}{e_{\rm B}}
\newcommand{\bHb}{\bH ^{\rm B}}
\newcommand{\PhiB}{\Phi ^{\rm B}}
\newcommand{\phiB}{\phi ^{\rm B}}

\newcommand{\EHt}{\mathcal{E}_{{\rm H},\delta_N}}
\newcommand{\EHtpm}{\tilde{\mathcal{E}}_{{\rm H} \pm}}
\newcommand{\uHt}{u_{{\rm H},\delta_N}}
\newcommand{\uHtp}{\tilde{u}_{{\rm H} +}}

\newcommand{\uHtpm}{\tilde{u}_{{\rm H} \pm}}
\newcommand{\eHt}{e_{{\rm H} , \delta_N}}
\newcommand{\eHtt}{\tilde{e}_{{\rm H}}}
\newcommand{\muHt}{\mu_{\delta_N}}
\newcommand{\HHt}{H_{{\rm mf},\delta_N}}
\newcommand{\EHN}{{\mathcal{E}}_{\rm H}^{\lambda'}}
\newcommand{\uHN}{u_{\rm H}^{\lambda'}}
\newcommand{\uHNp}{u_{{\rm H} +}^{\lambda'}}
\newcommand{\uHNm}{u_{{\rm H} -}^{\lambda'}}
\newcommand{\PhiBN}{\Phi^{{\rm B},\lambda'}}
\newcommand{\phiBN}{\phi^{{\rm B},\lambda'}}

\newcommand{\Vtp}{V_{\rm t} ^+}
\newcommand{\Vtm}{V_{\rm t} ^-}
\newcommand{\Vtpm}{V_{\rm t} ^{\pm}}
\newcommand{\tV}{V_{\delta_N}}

\newcommand{\chim}{\chi_-}
\newcommand{\chip}{\chi_+}
\newcommand{\chipm}{\chi_{\pm}}
\newcommand{\etam}{\eta_-}
\newcommand{\etap}{\eta_+}
\newcommand{\etapm}{\eta_{\pm}}
\newcommand{\HT}{\tilde{H}_1}

\newcommand{\Nmpe}{\N_- ^{\perp}}
\newcommand{\Nppe}{\N_+ ^{\perp}}

\newcommand{\Npe}{\N ^{\perp}}

\newcommand{\halfLN}{{\frac{L_N}{2}}}

\newcommand{\sym}{{\mathrm{sym}}}

\newcommand{\re}{\mathrm{Re}\, }

\newcommand{\I}{{\rm{i}}}
\newcommand{\D}{{\rm{d}}}
\newcommand{\dss}{\displaystyle}

\newcommand{\be}{\begin{equation}}
\newcommand{\ee}{\end{equation}}
\newcommand{\bea}{\begin{eqnarray}}
\newcommand{\eea}{\end{eqnarray}}

\def\real{{\mathbb{R}}}

\def\proba{{\rm I\kern -.18em P}}

\newcommand{\ie}{i.e.}

\newcommand{\RHS}{right-hand side\;\,}

\newcommand{\onehalf}{\frac{1}{2}}

\newcommand{\bra}[1]{\ensuremath{\langle{#1} |}}

\newcommand{\localized}{{\rm loc}}

\title[Interacting bosons in double wells]{Interacting bosons in a double-well potential : localization regime} 

\author[N. Rougerie]{Nicolas ROUGERIE}
\address{Universit\'e Grenoble Alpes \& CNRS,  LPMMC,  F-38000 Grenoble, France}
\email{nicolas.rougerie@lpmmc.cnrs.fr}

\author[D. Spehner]{Dominique SPEHNER}
\address{Universit\'e Grenoble Alpes \& CNRS, Institut Fourier \&  LPMMC,  F-38000 Grenoble, France\\
Present address: Departamento de Ingenier\'{\i}a Matem\'atica, Universidad de Concepci\'on, Chile}
\email{dspehner@ing-mat.udec.cl}

\date{March, 2018}

\begin{document}

\begin{abstract}
We study the ground state of a large bosonic system trapped in a symmetric double-well potential, letting the distance between the 
two wells increase to infinity with the number of particles. In this context, one should expect an interaction-driven
transition between a delocalized state (particles are independent and all live in both wells) and a localized state (particles are correlated, half of them live in each well). 
We start from the full many-body Schr\"odinger Hamiltonian in a large-filling situation where the on-site interaction and 
kinetic energies are comparable. When tunneling is negligible against interaction energy, we prove a localization estimate 
showing that the particle number fluctuations in each well are strongly suppressed.  The modes in which the particles condense
 are minimizers of nonlinear Schr\"odinger-type functionals. 
\end{abstract}

\maketitle

\setcounter{tocdepth}{2}
\tableofcontents

\section{Introduction}\label{sec:intro}
The Mott insulator/superfluid phase transition manifests itself by 
an interaction-driven drastic change in transport properties of a quantum system. Under conditions where the non-interacting 
system would be conducting, repulsive interactions can induce an insulating behavior if they dominate tunneling effects of electrons between ions in a crystal (in solid state systems), or of atoms between the wells of a magneto-optic trapping potential (in cold atomic gases).
Signatures of the transition have been observed experimentally in cold Bose gases trapped by periodic lattice potentials at low integer fillings (a few atoms per site)~\cite{BloDalZwe-08,Gre-etal-Blo-02,Ger-etal-Blo-06}. They include a sudden change in the fluctuations of the numbers of particles on each site and the relative phases at some critical value of the ratio between the tunneling and interaction energies~\cite{FisWeiGrFis-89}. 

In this paper, we mathematically investigate the  case of bosons confined in a double-well potential
 in a large filling situation, i.e. when one has many particles per well.
This situation corresponds to  current experiments in cold atom physics, the trapped atoms forming an
 external Bose-Josephson junction~\cite{Est-etal-Obe-08,Gro-etal-Obe-10,Rie-etal-Tre-10}. 
Like in the multiple-well case, one expects a transition between a delocalized and a localized regimes, the latter occurring
when the interactions between particles are stronger  than the energy needed for a particle to tunnel from one well into another. In the double-well situation this transition is not, however, expected to be a sharp transition. Instead, one expects for large atom numbers
a wide transition regime in which
the particle numbers and relative phase fluctuations change smoothly  (Josephson regime).

In theoretical studies of the Mott transition, it is customary to use a tight-binding approximation and work with a Hubbard
 model~\cite{FisWeiGrFis-89,Leggett_2001}. This relies on assuming that only the ground state of each potential well is 
occupied. At low filling (few particles per well), this is certainly reasonable for the interaction energy within one well
 will usually be smaller than the gap above the well's ground state energy. The physics is then reduced to particles
 hopping/tunneling between wells and subject to on-site interactions.
In a large filling situation, it is not so clear that one can rely on such a simplified
model: the interactions between particles on a given site can (and will) change the mode in which particles condense. Nevertheless, some conditions of applicability of the two-mode approximation have been
worked out~\cite{MilCorWriWal-97} and
the two-mode Hubbard model has been used extensively in the physics
literature to study  external Bose-Josephson junctions (see for instance~\cite{Ferrini2011})
and has been successful in explaining
experimental results with a hundred up to thousand atoms per well~\cite{Gati_Oberthaler_2007,Est-etal-Obe-08}. 
The problem of going beyond the Bose-Hubbard description,
which has been also considered in the physics literature (see for example~\cite{GarDouCar-11}), does not seem to have previously been studied from a mathematical standpoint.

We here start from the full many-body Schr\"odinger Hamiltonian 
\begin{equation}\label{eq:hamil depart}
H_N = \sum_{j=1} ^N \left( -\Delta_j + V_N (x_j) \right) + \frac{\lambda}{N-1} \sum_{1\leq i < j \leq N} w(x_i-x_j) 
\end{equation}
for $N$ interacting particles in $\real^d$ ($d=1,2,3$) and consider the large $N$ limit of its ground state in the case 
where $V_N$ is a symmetric double-well potential. As appropriate for bosons, we consider the action of $H_N$ on the symmetric tensor product space 
$$\gH ^N := \bigotimes_{\rm sym} ^N L ^2 (\R^d)\simeq L^2_{\rm sym} (\R ^{dN})$$
and study its lowest eigenvalue and associated eigenfunction.

The first sum in the Hamiltonian~\eqref{eq:hamil depart} describes the kinetic  and 
potential energies of the bosons in presence of the external trapping potential $V_N$, with 
$x_j \in \R^d$ and  $\Delta_j$ standing for the position of the $j$-th particle 
and the corresponding Laplacian. 
The second sum in \eqref{eq:hamil depart} describes interactions among the particles, assumed to be repulsive. 
The fixed coupling constant $\lambda >0$ is multiplied by a scaling factor of order $1/N$, in such a way that 
interactions have a leading order effect in the limit $N\to \infty$, while the ground state energy per 
particle remains bounded (mean-field regime). The choice of fixing the range of the potential (mean-field limit) is mostly out of simplicity. One should certainly expect our results to remain true in a dilute limit (see e.g.~\cite[Chapter~7]{Rougerie-LMU} or~\cite[Chapter~5]{Rougerie-hdr} for a discussion of the distinction). 

We will not aim at a great generality for the interaction potential $w$.
In what follows, we denote by $\hat{w}$ its Fourier transform, $\| w \|_\infty=\sup_{x\in\R ^d}| w(x)|$ its sup
norm and by $B(0,R)$ 
 a ball of $\real^d$ of radius $R$ centered at the origin.

\begin{assumption}[\textbf{The interaction potential}]\label{assumptions-w}\mbox{}\\
The interaction potential $w$ is  positive, of positive type, symmetric, and bounded with compact support:
\begin{equation}\label{eq:asum w}
w >  0, \quad \hat{w} \geq 0, \quad w(x) = w(-x),\quad  \| w \|_\infty < \infty, \quad \supp (w) \subset B(0,R_w)
\end{equation}
for some $R_w>0$.
\end{assumption}

It is well-known~(see \cite{Lewin-ICMP,LewNamRou-14,Rougerie-LMU,Rougerie-spartacus} and references therein) that if the one-body potential 
$V_N \equiv V$ in~\eqref{eq:hamil depart} does not depend on $N$ and $V (x)$ goes to infinity when 
$|x| \rightarrow \infty$,  the lowest eigenvalue $E(N)$ of the Hamiltonian~\eqref{eq:hamil depart}
is given in  the large $N$ limit by  
\begin{equation}\label{eq:Hartree lim ener}
\lim_{N\to \infty} \frac{E(N)}{N} = \eH ( \lambda)\;,
\end{equation}
where 
$\eH( \lambda)$ is the Hartree energy, \ie, the minimum of the functional
\begin{equation}\label{eq:Hartree func}
\EH^\lambda [u] = \int_{\R ^d} \big( |\nabla u | ^2 + V |u| ^2 \big) \D x
+ \frac{\lambda}{2} \iint_{\R ^d \times \R ^d} |u(x)| ^2 w(x-y) |u(y)| ^2 \D x \D y 
\end{equation}
under the constraint $\| u\|_{L^2}^2 = \int_{\R^d} |u|^2 =1$.
The Hartree functional $\EH^\lambda [u]$ is obtained from the energy 
of the mean-field state $u ^{\otimes N}$ describing $N$ independent particles  in the same state $u$
as 
$$\EH^\lambda [u] = \frac{1}{N} \langle u ^{\otimes N} | H_N | u ^{\otimes N} \rangle\;.$$
The rationale behind~\eqref{eq:Hartree lim ener} is that the ground state $\Psi_N \in \gH ^N$
of the $N$-body Hamiltonian $H_N$ roughly behaves as  
\begin{equation}\label{eq:Hartree heurist}
 \Psi_N \underset{N \to \infty}\approx ( \uH^\lambda ) ^{\otimes N} 
\end{equation}
where $\uH^\lambda$ is the minimizer of  the Hartree functional.
Note that the latter is, under Assumption~\ref{assumptions-w}, unique modulo a constant phase and can be chosen to be positive.

\medskip

The situation changes when the trapping potential $V_N$ in  the Hamiltonian~\eqref{eq:hamil depart} is allowed to depend on $N$, which is the case of interest in this work. We shall consider a model with a symmetric double-well potential:
\begin{equation}\label{eq:double well}
V_N (x) = \min \left\{ V\left( x - \mathbf{x}_N\right), V\left(x + \mathbf{x}_N\right) \right\}\;,
\end{equation}
where $V$ is a fixed radial potential and the two localization centers $\pm \mathbf{x}_N=\pm ( L_N/2,0,\cdots,0)$ are along the first coordinate axis. We can with our methods deal with rather general radial confining potentials, but shall for simplicity stick to the model case of a power-law potential:
\begin{equation}\label{eq:hom pot}
V (x) = |x| ^s,\quad s \geq 2 \; .  
\end{equation}
To mimic a potential with two deep and well-separated wells, we let the inter-well distance 
$$L_N = 2 |\mathbf{x}_N| \to \infty$$
in the limit $N\to \infty$.  By scaling, this situation is equivalent to the one where the distance $L_N$ stays fixed  and the range of the interaction goes to $0$, but the potential barrier $V_N (0)$ goes to infinity.

In the following,  the coupling constant $\lambda$  is kept fixed and we will
often omit it in the upper index to simplify notation. We denote by $\EHm [u]$ and $\EHp [u]$
the functionals~\eqref{eq:Hartree func} in which  $V_N$ is replaced by the  left and right
 potential wells, given respectively by
$$ 
V_N ^- (x) = V\left(x + \mathbf{x}_N \right), \quad V_N ^+ (x) = V\left(x - \mathbf{x}_N\right)\;.
$$ 
Obviously, the positive minimizers  $\uHpm$ of $\EHpm [u]$ are equal to the same   
fixed function $\uH$ modulo a translation by  $\pm \mathbf{x}_N$ and the corresponding Hartree energies are 
equal,
$$
\EHpm [ \uHpm ] := \inf_{ \| u\|_{L^2} = 1} \big\{ \EHpm [ u] \big\}
 =  \eH \;. 
$$ 
When the inter-well distance $L_N$ is very large, $\uHm$ and $\uHp$ become almost  orthogonal to one another since they are localized in far-apart potential wells.

\medskip

In the case of a single particle with Hamiltonian $H_1=-\Delta + V_1$, it is well-known that 
the lowest energy state $\Psi_1$ is close in the limit $L_1 \to \infty$ to the symmetric superposition
\begin{equation}\label{eq:ansatz deloc 1p}
\Psi_1 \approx C_{\mathrm{dloc}} ^1 \left( \frac{u_{-} +u_{+}}{\sqrt{2}}  \right)\;,
\end{equation}
where $u_{\pm}$ is the ground state of $-\Delta + V_1^\pm$ and $C_{\mathrm{dloc}} ^1 $ is a normalization factor.
More precisely, denoting by $e$ the lowest eigenvalue of the  Hamiltonian $-\Delta + V$ in a single well,
 it can be shown~\cite{AveSeil-75,Daumer-96,Davies-82,ComSei-78,Harrell-78,Harrell-80,HelSjo-84,HelSjo-85,MorSim-80} that the spectrum of $H_1$ in an interval of length of order one
centered around $e$ consists of exactly two eigenvalues,
which converge to $e$ as $L_N \rightarrow  \infty$, and that the eigenfunction $\Psi_1$ associated to the lowest of these
eigenvalues satisfies 
$$\norm{\Psi_1 - \frac{u_{-} + u_{+} }{\sqrt{2}}} \to 0$$ 
as $L_1 \to \infty$. This embodies the fact that a particle in the ground state of $H_1$ has equal probabilities of being in the left or right well, \ie, the ground state~\eqref{eq:ansatz deloc 1p} is delocalized over the two wells. Extensions of this result to nonlinear models are given in~\cite{Daumer-91,Daumer-94}.

\medskip

When there is more than one particle, the situation may change completely due to the repulsive interactions. If both $N$ and $L_N$ are large, one should expect a transition between: 

\medskip

\noindent $\bullet$ A regime where a delocalized ground state akin to~\eqref{eq:Hartree heurist} is preferred,
occurring when  $L_N$ is not too large.
Actually, if $L_N$ is small enough so that tunneling dominates over interactions,
a reasonable approximation is to replace the Hartree minimizer in~\eqref{eq:Hartree heurist} by a symmetric superposition, in analogy with \eqref{eq:ansatz deloc 1p}. This leads to the heuristic
\begin{equation}\label{eq:ansatz deloc}
\Psi_N \approx  \Psidloc = C_{\mathrm{dloc}} ^N \left( \frac{\uHm + \uHp }{\sqrt{2}} \right) ^{\otimes N}  
\end{equation}
with the normalization factor
$$C_{\mathrm{dloc}}^N = \left( 1 + \left\langle \uHm , \uHp \right\rangle\right)^{-N/2} \to 1$$ 
when $L_N \to \infty$. In the $N$-body state $\Psidloc$, all particles are independent and identically distributed in the same quantum state, delocalized over the two wells. 

\smallskip

\noindent $\bullet$ A regime where a localized state emerges to reduce on-site interactions, occurring
for larger inter-well distances $L_N$. 
An ansatz for such a state can be taken of the  form (hereafter we assume that $N$ is even)\footnote{The definition of the symmetrized
tensor product $\otimes_{\mathrm{sym}}$ is recalled below, see~\eqref{eq-symm_tensor_prod}.}
\begin{equation}\label{eq:ansatz loc}
\Psi_N \approx \Psiloc = C_{\mathrm{loc}}^N \, \uHm^{\otimes N/2} \otimes_{\mathrm{sym}} \uHp^{\otimes N/2}\;.
\end{equation}
As above,  the normalization factor  $C_{\mathrm{loc}}^N \to 1$  when $L_N \to \infty$, with small corrections of the order of $\langle \uHm , \uHp \rangle$.
The ansatz $\Psiloc$ is a correlated state where  half of the particles live in the left well $V_N^-$ and 
the other half in the right well $V_N^+$. Note that the ansatz~\eqref{eq:ansatz loc} involving two one-body wave functions has a kinship with states used to describe two component Bose-Einstein condensates, see e.g.~\cite{AnaHotHun-17,MicOlg-16,MicOlg-17,Olgiati-17}.
The physics is however very different, and so shall our analysis be.

\medskip

In this paper, we focus on the regime where localization prevails. Note that this should \emph{not} be interpreted as~\eqref{eq:ansatz loc} being very close to the true ground state, even in the sense of reduced density matrices (see Remark~\ref{rem:DMs} and Section~\ref{sec:heuristic2}
below). The simple ansatz~\eqref{eq:ansatz loc} is in fact motivated by what happens very deep in the localization regime. In the regime close to the transition, which is our concern here, localization is not as strong as in~\eqref{eq:ansatz loc} and one should be careful about what it actually means.

We formulate localization as follows. Denote by $a^\ast(u)$ and $a(u)$ the bosonic creation and annihilation operators\footnote{The definition is recalled below, see~\eqref{eq:anni cre}.} in a mode $u \in L^2 ( \R^d)$. Let
\begin{equation}\label{eq:part num op}
\N_- = a^\ast (\uHm ) a(\uHm ) \: \mbox{ and } \: \N_+ = a^\ast (\uHp ) a(\uHp ) 
\end{equation}
be the number operators in the modes $\uHm$ and $\uHp$ localized in the left and right wells, respectively. We say that the system is localized if, in the large $N$ limit, the variance of $\N_\pm$ satisfies\footnote{
  By symmetry of the potential $V_N$, it is easy to see that the expectation of $\N_\pm$ in the ground state $\Psi_N$ is $\bra{\Psi_N} \N_{\pm} \ket{\Psi_N} = N/2$. Thus the quantity in the
  left-hand side of~\eqref{eq:part num op} is the variance of $ \N_\pm$.
}  
\begin{equation}\label{eq:intro main result}
\boxed{\left\langle \Psi_N \left|  \left(\N_\pm - \frac{N}{2}\right) ^2 \right| \Psi_N \right \rangle \ll N \;}
\end{equation}
with $\Psi_N$ the ground state of the many-body Hamiltonian. Then the fluctuations of $\N_-$ and $\N_+$  are reduced with respect to the case of independent particles~\eqref{eq:ansatz deloc}, where they would be of order $N$  since 
$$\left\langle u^{\otimes N} \, \big|\, \N_\pm^2 \, \big|\, u^{\otimes N}\right\rangle  - \left\langle u^{\otimes N} \, \big|\, \N_\pm \, \big|\, u^{\otimes N} \right\rangle^2  = N \left\langle \uHpm , u \right\rangle$$ 
for any  $u \in L^2 ( \R^d)$. The reduced fluctuations in~\eqref{eq:intro main result} constitute a violation of the central limit theorem and show the occurrence of strong correlations, akin to those of~\eqref{eq:ansatz loc}, in the ground state of the system. 

Our main result in this paper shows that, for a fixed $\lambda$, localization in the sense of~\eqref{eq:intro main result} occurs when $N\to \infty$ and $L_N \rightarrow \infty$ satisfy
\begin{equation}\label{eq:loc regime}
\log N  \leq 2 ( 1- \eps) A \Big(\frac{L_N}{2} \Big)
\end{equation}
for some arbitrarily small fixed $\eps >0$, where  
\begin{equation}\label{eq:Agmon}
A(r)=  \int_0^{r} \sqrt{V(r')} \D r' 
\end{equation}
is the Agmon distance (at zero energy) from semiclassical analysis~\cite{Agmon-82}. In the model case~\eqref{eq:hom pot} we have 
\begin{equation}\label{eq:Agmon hom} 
A (r) = \left( 1+ \frac{s}{2} \right) ^{-1} r ^{1 + s/2}. 
\end{equation}

Note that, although the localized and delocalized states $\Psiloc$ and $ \Psidloc$ in 
\eqref{eq:ansatz deloc} and \eqref{eq:ansatz loc} 
have very different physical properties, distinguishing them in 
the large $N$ limit is not as easy as one might think.
Actually, as we shall see in Section~\ref{sec:heuristic}, the difference between the interaction energies per particle in the states $\Psiloc$ and $ \Psidloc$ is of order $1/N$.
This is of the same order as  the next-to-leading order term in the large $N$ expansion  of the ground state energy in a single well, due to Bogoliubov fluctuations~\cite{BocBreCenSch-17,CorDerZin-09,Seiringer-11,GreSei-13,LewNamSerSol-13,DerNap-13,NamSei-14,LewNamSerSol-13}. Indeed, if the potential $V_N \equiv V$ in~\eqref{eq:hamil depart} is independent of $N$, one can go beyond \eqref{eq:Hartree lim ener} and prove that
\begin{equation}\label{eq:Hartree lim 2}
\frac{E(N)}{N} = \eH (\lambda) + N ^{-1} \eB (\lambda) + o (N ^{-1})\;,
\end{equation}
where 
$\eB (\lambda)$ is the ground state energy of the Bogoliubov Hamiltonian, obtained from $H_N$ by a suitable expansion around the condensed state $\uH ^{\otimes N}$ (see Section~\ref{sec:bog} below for more details). We will prove that Bogoliubov fluctuations, even though they must be taken into account in the  analysis of the problem,  do not play an important role in deciding which of the localized and delocalized states has the smallest energy.

\bigskip 

\noindent \textbf{Acknowledgments.} Warm thanks to Phan Th\`anh Nam for many detailed discussions on Bogoliubov's theory. We were financially supported by the French ANR project ANR-13-JS01-0005-01.
\section{Main results and discussion}\label{sec:main res}
\subsection{Heuristics} \label{sec:heuristic}

The order of magnitude (as a function of $N$) of the inter-well distance $L_N$ at which the localized state
(\ref{eq:ansatz loc}) has a lower energy than the delocalized state (\ref{eq:ansatz deloc})
can be derived heuristically as follows.
%
Let us consider
the \emph{tunneling energy}
\begin{align}\label{eq:tunneling term}
\nonumber 
T_N  &= \left\langle \uHm |  ( -\Delta + V_N ) | \uHp \right\rangle  \\
&= \bigg\langle \frac{\uHm+\uHp}{\sqrt{2}} \bigg| (  - \Delta + V_N  ) \bigg|  \frac{\uHm+\uHp}{\sqrt{2}} \bigg\rangle \nonumber\\
 &- \frac{1}{2} \langle \uHm |  ( - \Delta + V_N ) | \uHm\rangle - \frac{1}{2} \langle \uHp |  ( - \Delta + V_N ) | \uHp \rangle\;. 
\end{align}
Recall the variational equations satisfied by $\uHpm$:
\begin{equation} \label{eq:var eq Hartree}
\left[ -\Delta + V_N^\pm  + \lambda w*|\uHpm | ^2 \right]  \uHpm = \mu \uHpm  
\end{equation}
with $*$ denoting convolution and $\mu$ the chemical potential (Lagrange multiplier), 
\begin{equation} \label{eq-chem pot}
\mu = \eH + \frac{\lambda}{2} \iint_{\R^d \times \R^d} |\uHpm (x)| ^2 w(x-y) |\uHpm (y)| ^2 \D x \D y
\;.
\end{equation}
Inserting~\eqref{eq:var eq Hartree} in~\eqref{eq:tunneling term} we obtain 
\begin{equation}\label{eq:tunneling termbis}
T_N = \int_{\R ^d} \uHm \Vtp \uHp  = \int_{\R ^d} \uHp \Vtm \uHm \;,
\end{equation}
where $\Vtpm$  are the tunneling potentials
\begin{equation}\label{eq:tunneling pot}
\Vtpm = V_N - V_N ^\pm  - \lambda w * |\uHpm | ^2   + \mu\;.
\end{equation}
From~\eqref{eq:tunneling termbis} one can derive that $T_N \leq 0$ for large enough inter-well distances $L_N$
(see Proposition~\ref{cor:tunneling} below). Going back to~\eqref{eq:tunneling term},
this is equivalent to the symmetric delocalized state
$(\uHp + \uHm)/\sqrt{2}$ having a lower one-body (\ie, kinetic and potential)
energy than the states $\uHm$ and $\uHp$ localized
in the left and right wells.

Actually, we recall that the energy of $N$ bosons in a state $\Psi \in \gH ^N $ is given by 
\begin{equation}\label{eq:ener DM}
 \langle \Psi | H_N | \Psi \rangle 
= E^{({\rm kin+pot})}_\Psi + E^{({\rm int})}_\Psi \;,
\end{equation}
where the energy components are 
\begin{equation} \label{eq-energy_with_density_matrices}
E^{({\rm kin+pot})}_\Psi =
\tr \bigl[ ( -\Delta + V_N ) \gammaP ^{(1)} \bigr] 
\quad , \quad 
E^{({\rm int})}_\Psi
=   \frac{\lambda}{2(N-1)} \tr \bigl[  w \, \gammaP ^{(2)} \bigr] \;
\end{equation}
and the $k$-body density matrices $\gammaP^{(k)}$ are  defined by
\begin{equation}\label{eq:def DM}
\gammaP ^{(k)} =  \frac{N!}{(N-k)!} \tr_{k+1\to N} \left[ \,|\Psi \rangle \langle \Psi | \,\right]
\end{equation}
or, equivalently~\cite[Section 1]{Lewin-11},
\begin{equation} \label{eq:other def DM}
\left\langle v_1 \otimes_{\mathrm{sym}} \cdots \otimes_{\mathrm{sym}} v_k , \gammaP^{(k)}  u_1 \otimes_{\mathrm{sym}} \cdots \otimes_{\mathrm{sym}} u_k \right\rangle
= k! \left\langle \Psi | a^\ast ( u_1) \cdots a^\ast (u_k) a(v_1)\cdots a(v_k) | \Psi \right\rangle
\end{equation}
for any $k=1,\cdots, N$ and $u_1,v_1,\cdots,u_k,v_k \in \gH$.
Hereafter,  the symmetric tensor product $\otimes_{\mathrm{sym}}$  is defined   by\footnote{Note that
$\Psi_1 \otimes_\mathrm{sym} \Psi_2$ is not normalized even if this is the case for $\Psi_1$ and $\Psi_2$, for instance
$u \otimes_\mathrm{sym} u = \sqrt{2} u^{\otimes 2}$.}
\begin{multline} \label{eq-symm_tensor_prod}
 \Psi_1 \otimes_\mathrm{sym} \Psi_2 (x_1,\ldots, x_{N}) :=  \frac{1}{\sqrt{N_1! N_2! N!}} \\ 
 \sum_{\sigma \text{ permutation of } \{ 1, \ldots, N\}} \Psi_1 ( x_{\sigma(1)} , \ldots, x_{\sigma(N_1)})
 \Psi_2  ( x_{\sigma(N_1+1)} , \ldots, x_{\sigma(N)}) 
 \end{multline}
for any  $\Psi_1 \in \gH ^{N_1}$ and  $\Psi_2 \in \gH ^{N_2}$ with $N = N_1 + N_2$.

A simple calculation shows that 
 the 1-body density matrix $\gammaidloc$
 and $\gammailoc$ in the delocalized and localized states  $\Psidloc$ and $\Psiloc$
are given by
\begin{align}
\gammaidloc &\simeq \frac{N}{2} \left( |\uHm \rangle \langle \uHm | + |\uHp \rangle \langle \uHp | +  |\uHm \rangle \langle \uHp | + |\uHp \rangle \langle \uHm |  
\right) \label{eq:DM 1 dloc} \\
\label{eq:DM 1 loc}
\gammailoc &\simeq \frac{N}{2} \left(  |\uHm \rangle \langle \uHm | + |\uHp \rangle \langle \uHp | \right) , 
\end{align}
up to small corrections of order $N \langle \uHm \,,\, \uHp\rangle$. These can be neglected in the limit $L_N \rightarrow \infty$. One then infers from~\eqref{eq:ener DM},~\eqref{eq:DM 1 dloc}, and~\eqref{eq:DM 1 loc} that
\begin{equation}
 E^{({\rm kin+pot})}_{\rm dloc} - E^{({\rm kin+pot})}_{\rm loc} 
= N T_N < 0\;.
\end{equation}
As a result of tunneling between the two wells, $\Psidloc$ has a lower one-body energy  than  $\Psiloc$.  

On the other hand, $\Psiloc$ has a lower interaction energy than  $\Psidloc$. 
Indeed, it is easy to see that the 2-body density matrices of both states
have ranges in the 3-dimensional subspace with basis $\{ \ket{\uHp^{\otimes 2}}, \ket{\uHm^{\otimes 2}}, \ket{\uHm \otimes_\sym \uHp} \}$.
Since $\uHm$ and $\uHp$ are well-separated in space and $w$ is short-ranged, one can neglect in the limit $L_N \to \infty$
all the matrix elements of $w$ in this basis save for the two elements 
\begin{eqnarray*}
  \langle  \uHpm^{\otimes 2} , w \, \uHpm^{\otimes 2} \rangle = \langle \uH^{\otimes 2}  , w  \, \uH^{\otimes 2} \rangle =
  \iint_{\R^d \times \R^d} | \uH |^2 (x) w (x-y)  | \uH |^2(y) \D x \D y 
\end{eqnarray*}
corresponding to on-site interactions (see Remark~\ref{rem:bound_int_energy_term} below).
The first equality follows from the translation invariance and parity of $w$. 
By using~\eqref{eq:part num op}, (\ref{eq:other def DM}), and the commutation relations of $a$ and $a^\ast$ one finds
\begin{eqnarray*}  
\langle \uHpm^{\otimes 2}, \gamma_\Psi^{(2)} \uHpm^{\otimes 2} \rangle = \bra{\Psi} \N_\pm ( \N_\pm -1 ) \ket{\Psi} \;. 
\end{eqnarray*}
We infer from \eqref{eq-energy_with_density_matrices} that
\begin{equation} \label{eq-interaction_energy_general_Psi}
  E_\Psi^{({\rm int})} =  \frac{\lambda}{2(N-1)} \langle \uH^{\otimes 2} , w  \, \uH^{\otimes 2} \rangle
  \bra{\Psi} \left(  \N_+ ( \N_+ -1 ) + \N_{-} ( \N_{-} -1 )  \right) \ket{\Psi}\;.
\end{equation}
The interaction energies of the delocalized and localized 
states are thus given by
\begin{equation}
E_{\rm dloc}^{({\rm int})} =
\Bigl( 1 + \frac{1}{N-2} \Bigr)  E_{\rm loc}^{({\rm int})} 
= \frac{\lambda N}{4}
\langle \uH^{\otimes 2} , w \, \uH^{\otimes 2} \rangle
\end{equation}
The localized state thus favors the interaction energy, but only by a small amount, $O(N ^{-1})$ in the energy per particle. 

We deduce from this discussion that the 
limits of larges $N$ and $L_N$ (with fixed $\lambda$)
for which the localized state (\ref{eq:ansatz loc}) has a lower energy than
the delocalized state is given by 
\begin{equation}\label{eq:loc regime pre}
\text{ Localization regime: } N\to \infty, \quad  ,\quad \lambda N ^{-1} \gg |T_N| 
\end{equation}
We warn the reader that while this limit is obtained by comparing the energies of
the two ground states $\Psiloc$ for zero tunneling and $\Psidloc$ for vanishing interactions,
the true ground state of $H_N$ differs significantly
from both $\Psiloc$ and $ \Psidloc$ when $\lambda N^{-2} \ll |T_N| \ll \lambda$
(see Remark~\ref{rem:DMs} and Section~\ref{sec:heuristic2} below).
Although a better definition would be given by the
localization criterion \eqref{eq:intro main result},
 we hereafter refer to the limit~\eqref{eq:loc regime pre} as the ``localization regime'' since
 we are able to prove~\eqref{eq:intro main result} in this limit. We do not claim optimality however, see the 
better estimates of the ground state in Section~\ref{sec:heuristic2} below.

As we shall later see, due to the presence of the non-linearity we are 
not able to  evaluate exactly the order of magnitude of $|T_N|$, but we get in  Proposition~\ref{cor:tunneling}
a rather precise estimate: for any $\varepsilon >0$, 
\begin{equation}\label{eq:estim tunnel}
 c_\varepsilon \exp\left( -2 (1 + \varepsilon)  A \Bigl( \frac{L_N}{2} \Bigr) \right) \leq |T_N| 
 \leq C_\varepsilon \exp\left( -2 (1 - \varepsilon)  A \Bigl( \frac{L_N}{2} \Bigr)  \right) \;,
\end{equation}
where $A(r)$ is the Agmon distance~\eqref{eq:Agmon} associated with the single-well potential $V$ and $c_\varepsilon$ and $C_\varepsilon$ are positive constants depending only on~$\varepsilon$.
In view of~\eqref{eq:estim tunnel} and since $A(L_N/2)\to \infty$ as $L_N \to \infty$, 
the localization condition~\eqref{eq:loc regime pre} is satisfied when for any $\varepsilon >0$ and fixed $\lambda$,
\begin{equation}\label{eq:loc regime 2}
\boxed{\text{Localization: } N\to \infty, \quad L_N \rightarrow \infty , \quad \log N  \leq 2 ( 1 - \varepsilon) A \left( \frac{L_N}{2} \right).} 
\end{equation}

 In the sequel,
 when we will write that~\eqref{eq:loc regime} or~\eqref{eq:loc regime 2} holds, this will always mean that it does so for some 
$\varepsilon >0$ that one can choose arbitrarily small, independently of $N$.

\subsection{Main theorem}

One difficulty is apparent from the previous discussion: we are trying to capture a transition governed by a correction of order $N^{-1}$ to the ground state energy per particle. On-site fluctuations  are responsible for another correction of the 
same order of magnitude, cf~\eqref{eq:Hartree lim 2}. The intuition discussed above is nevertheless correct and a localized state will  be preferred in the 
regime~\eqref{eq:loc regime 2}. A rigorous proof of this fact requires a detailed analysis taking into account on-site Bogoliubov fluctuations.

To state our main result we recall that the single-well Hartree energy $\eH(\lambda)$ at coupling constant $\lambda$ 
is defined as the minimum of the energy functional~\eqref{eq:Hartree func} with $V$ the single-well potential~\eqref{eq:hom pot}. 
The Bogoliubov energy $\eB (\lambda)$ is obtained as the lowest eigenvalue of the second quantization of the Hessian of $\EH^\lambda$ around its minimum, see Section~\ref{sec:bog} for details.

\begin{theorem}[\textbf{Localized Regime}]\label{thm:main loc}\mbox{}\\
Let $\lambda \geq 0$ be a fixed constant. In the limit~\eqref{eq:loc regime}, we have
\begin{itemize}
\item \emph{(Energy asymptotics):} the ground state energy $E(N)$ of $H_N$ satisfies
\begin{equation}\label{eq:ener lim loc}
\boxed{\frac{E(N)}{N} = \eH  \left( \Delta_N \frac{\lambda}{2} \right) + \frac{2}{N} \eB\left( \frac{\lambda}{2} \right) + o (N ^{-1})} 
\end{equation}
where
\begin{equation}\label{eq:delta N}
\Delta_N = 1 - \frac{1}{N-1}. 
\end{equation}
\item \emph{(Particle number fluctuations):} the ground state  $\Psi_N$ of $H_N$ satisfies 
\begin{equation}\label{eq:sup fluctu}
\boxed{\left\langle  \Psi_N \left| \left(\N_+ - \frac{N}{2}\right) ^2 \right| \Psi_N \right\rangle
 +  \left\langle  \Psi_N \left| \left(\N_- - \frac{N}{2}\right) ^2 \right| \Psi_N \right\rangle \ll N}
\end{equation}
where $\N_\pm$ are the particle number operators in the left and right wells, defined in~\eqref{eq:part num op}.
\end{itemize}
\end{theorem}

\begin{remark}[Composition of the energy]\label{rem:energy}\mbox{}\\
  The energy expansion~\eqref{eq:ener lim loc} is a first signature of a transition to a localized state.
  It coincides (up to errors of order $o(N^{-1})$)
  with the ground state energy of two independent bosonic gases localized infinitely far apart
  in the left and right wells, having $N/2$ particles each.
  Note that, since the coupling constant in the Hamiltonian \eqref{eq:hamil depart} is  $\lambda (N-1)^{-1}$ instead of
  $\lambda (N/2-1)^{-1}$, the parameter $\lambda$ should be renormalized as $\lambda \rightarrow \lambda \Delta_N/2$,
  so that the ground state energy of each gas is equal to
  $$
  (N/2) \eH  ( \Delta_N \lambda/2 ) + \eB ( \lambda/2) + o(1).$$
  The energy in the right-hand side of~\eqref{eq:ener lim loc} is therefore
 equal to the sum of the
 lowest energies of the two gases in the left and right wells, up to errors of order $o(1)$. 
  Indeed,~\eqref{eq:ener lim loc} can 
  be obtained using as trial state a refinement of~\eqref{eq:ansatz loc}, taking into account on-site Bogoliubov fluctuations.

In the regime~\eqref{eq:loc regime} one can see that the delocalized ansatz~\eqref{eq:ansatz deloc}, supplemented by the appropriate Bogoliubov fluctuations, has a larger energy per particle, by an  amount $O(N ^{-1})$. Indeed, the first term in~\eqref{eq:ener lim loc} then becomes $\eH  \left( \frac{\lambda}{2} \right)$ while the second one is unchanged at leading order\footnote{This is not so easy to show, see the expressions in~\cite{GreSei-13}.} and the tunneling contribution is negligible. \hfill$\diamond$
\end{remark}

\begin{remark}[Strong correlations]\label{rem:correlation}\mbox{}\\
To appreciate that~\eqref{eq:sup fluctu} is a signature of correlations, the following considerations are helpful. First note that $\N_-$ and $\N_+$ could be seen as random variables, and that, roughly speaking, 
$$ \N_- = \sum_{j=1} ^N X_j$$
where $X_j$ is a random variable which takes the value $1$ if particle $j$ is in the $-$ well and~$0$ if it is in the $+$ well. Each $X_j$ has mean $1/2$ so that 
$$ \left\langle  \Psi_N \big| \N_- \big| \Psi_N \right\rangle = \frac{N}{2}.$$
Now, if correlations between particles could be neglected, the $X_j$'s would be independent random variables. Using the central limit theorem, one would expect the variance 
$$ \left\langle  \Psi_N \left| \left(\N_- - \frac{N}{2}\right) ^2 \right| \Psi_N \right\rangle$$
to scale like $N$ when $N\to \infty$. Our result~\eqref{eq:sup fluctu} rules this out, and thus the $X_j$'s cannot be independent. Note that weakly correlated particles usually also satisfy central limit theorems, see for example~\cite{BenKirSch-13,BucSafSch-13}. Thus,
\eqref{eq:sup fluctu} implies that  the bosons
in the double-well potential must be strongly correlated in the localized regime. In contrast, in the delocalized regime
one expects weak correlations, and thus
particle number fluctuations of the order of $\sqrt{N}$.
\hfill$\diamond$

\end{remark}

\begin{remark}[More general single-well potentials $V$]\mbox{}\\
As mentioned in the Introduction, the results of
Theorem~\ref{thm:main loc} and the estimate \eqref{eq:estim tunnel} on the tunneling energy $T_N$ are in fact
valid for more general single-well potentials $V(x)$, not necessarily given by power laws.
For instance, one can easily generalize
all the estimates of Section~\ref{sec:hartree} and the proof in the subsequent sections to radial potentials $V(r)$ 
satisfying the following assumptions:
\begin{itemize}
\item[(a)]
$V(r) \geq 0$ and $V(r)$ is increasing on $(r_0,\infty)$ for some fixed $r_0>0$;
\item[(b)]
$\dss\lim_{r \to \infty}  \frac{\D}{\D r} \sqrt{ V(r)}$ exists and belongs to $(0,\infty]$;
\item[(c)]
$V'(r)/V''(r) \rightarrow \infty$, 
$V(r) /V'(r) \rightarrow \infty$ and $A(r)/\sqrt{V(r)} \rightarrow \infty$
when $r \rightarrow \infty$, with $A(r)$ the Agmon distance~\eqref{eq:Agmon} associated to $V$;
\item[(d)] $\dss \lim_{r \to \infty} \frac{r V'(r)}{V(r)} = s$ exists and belongs to $[2,\infty]$.
\end{itemize}
\hfill$\diamond$
\end{remark}

\begin{remark}[Ground state in the localization regime] \label{rem:DMs} \mbox{}\\
  The heuristic arguments presented in Section~\ref{sec:heuristic} to identify the localization regime
are very rough since we have merely compared the energies of two states $\Psiloc$ and $\Psidloc$, corresponding
respectively  to the true ground states in the absence of tunneling ($L_N = \infty$) and for vanishing interactions ($\lambda=0$).
 It turns out that
 the ground state of the many-body Hamiltonian \eqref{eq:hamil depart}
is not close to the purely localized ansatz~\eqref{eq:ansatz loc} in the whole
localization regime $|T_N| \ll \lambda N^{-1}$.
In fact, its one-body density matrix $\gamma_{\Psi_N}^{(1)}$ does not even have two macroscopic eigenvalues of order $N$, and a fortiori 
is not close to the diagonal density matrix $\gamma_{\rm loc}^{(1)}$ given by \eqref{eq:DM 1 loc}.
As we will see in the next subsection and Appendix~\ref{app:BH}, closeness to  $\gamma_{\rm loc}^{(1)}$ should be expected to hold
for lower tunneling energies  $|T_N| \ll \lambda N^{-2}$ only. For higher $|T_N|$,  $\gamma_{\Psi_N}^{(1)}$ is instead expected to be
close to the density matrix $\gamma_{\rm dloc}^{(1)}$  of the {\emph{delocalized}} state, given by \eqref{eq:DM 1 dloc}.

Even if Theorem~\ref{thm:main loc} does not provide a full characterization of the ground state,
it captures its most physically important feature, namely the
reduced  fluctuations of particle numbers in each well (squeezing),
which implies as mentioned before the presence of strong correlations between particles
(such correlations are of course not seen in the one-body density matrix).
One may conjecture from heuristic
arguments (see the next subsection) that this property holds more generally
in the limit $N \to \infty$, $L_N \to \infty$, $\lambda$ fixed, \ie, 
it also occurs for smaller inter-well distances $L_N$ which do not satisfy~\eqref{eq:loc regime}.
Proving this is, however, out of reach from the methods presented in Section~\ref{sec:low bound loc}. 
\hfill$\diamond$
\end{remark}

\subsection{More precise heuristics} \label{sec:heuristic2}

The properties of the ground state of interacting bosons in a symmetric double-well potential
have been studied extensively in the physics literature (see e.g. the review articles~\cite{Gati_Oberthaler_2007,Leggett_2001}).
We summarize them in Table~\ref{tab1} and derive them heuristically in this subsection and in Appendix~\ref{app:BH}, neglecting on-site Bogoliubov fluctuations as in Section~\ref{sec:heuristic}. 

The \emph{main conjectures} we wish to argue for in this subsection are that
\begin{itemize}
\item localization in the sense of~\eqref{eq:sup fluctu} holds when the tunneling energy satisfies
  $|T_N| \ll \lambda$  (compare with~\eqref{eq:loc regime pre}). 
\item this is essentially sharp, i.e.~\eqref{eq:sup fluctu} fails for $|T_N| \gg \lambda$. 
\end{itemize}
Proving these conjectures remains out of reach of our present method, for this would require much finer estimates of the tunneling contribution to the ground-state energy.
Note that the first conjecture implies that
localization  in the sense of~\eqref{eq:sup fluctu} always occurs  in the limit $N \to \infty$, $L_N \to \infty$, $\lambda$ fixed
(in fact, one has $T_N \to 0$ as $L_N \to \infty$).

\begin{table}[h]
\begin{tabular}{|c||c|c|c|}
\hline
Limit                             &  \begin{tabular}{c}  {\it Fock regime}  \\ $N \to \infty , |T_N| \ll \lambda N^{-2}$ \end{tabular}
& \begin{tabular}{c} {\it Josephson regime} \\ $N \to \infty, \lambda N^{-2} \ll |T_N| \ll \lambda$ \end{tabular} 
&  \begin{tabular}{c} {\it Rabi regime} \\ $N \to \infty , \lambda \ll |T_N|$ \end{tabular}
\\[1mm]
\hline
\begin{tabular}{c}  expected\\  ground state \end{tabular}
&  Fock state $\Psiloc$    &  Squeezed state $\Psi_{\rm sq}$    &   Coherent state $\Psidloc$
\\[1mm]
\hline
\begin{tabular}{c}  particle number\\  fluctuations \end{tabular}
&  $\langle (\Delta \N_{-})^2 \rangle = O (1)$ & $\langle (\Delta \N_{-})^2 \rangle \ll N$
& $\langle (\Delta \N_{-})^2 \rangle = O (N)$
\\[1mm]
\hline
 tunneling factor    & $\langle \uHm, \gamma^{(1)}_{\Psi_N} \uHp \rangle = O(1)$     &   $\langle \uHm, \gamma^{(1)}_{\Psi_N} \uHp \rangle \approx \frac{N}{2}$   & $\langle \uHm, \gamma^{(1)}_{\Psi_N} \uHp \rangle \approx \frac{N}{2}$
\\[1mm]
\hline
\end{tabular}

\bigskip

\caption{\label{tab1}
  Expected properties of the ground state of the many-body Hamiltonian \eqref{eq:hamil depart} for large $N$ and $L_N$ (see e.g.~\cite{Leggett_2001,Gati_Oberthaler_2007}).
  Here $|T_N|$ is the tunneling energy, decaying with $L_N$ roughly as $e^{-2 A(L_N/2)}$, where $A(r)$ is the Agmon distance (see \eqref{eq:estim tunnel}), and $\lambda N^{-1}$ is the coupling constant for inter-particle interactions. 
  We prove rigorously in this paper the reduced particle number fluctuations in the limit $N \to \infty$, $|T_N| \ll N^{-1}$, $\lambda$ fixed,
  that is,
  from the Fock regime up to the middle of the Josephson regime.}
\end{table}

Instead of investigating the many-body Hamiltonian $H_N$, most studies in the physics literature
deal with the simpler two-mode Bose-Hubbard  Hamiltonian, which is obtained by restricting $H_N$ to the subspace $\gH_{\rm BH} \subset \gH ^N$ spanned by the $N+1$ Fock states
\begin{equation*}
\ket{n,N-n} = \frac{1}{\sqrt{n! (N-n)!}} (a_{-}^\ast)^n (a_+^\ast)^{N-n} \ket{0} =  
C_n \uHm^{\otimes n} \otimes_\mathrm{sym} \uHp^{\otimes (N-n)} 
\end{equation*}
where $n=0,\cdots, N$, $\ket{0}\in \gH ^N$ denotes the vacuum state,
$a_{-} = a ( \uHm)$ and $a_{+} = a ( \uHp)$ are the annihilation operators in the states
$\uHm$ and $\uHp$ minimizing the Hartree functionals in the left and right wells, and $C_n$ is a normalization factor\footnote{
  A few (mainly numerical) works in the physics literature go beyond the two-mode approximation.
  For instance, perturbative and exact diagonalization approaches have been used in Ref.~\cite{GarDouCar-11}   
  to include also the first excited state in each well.
  }.
The energy of a general state in $\gH_{\rm BH}$, 
\begin{equation} \label{eq-general_states_Bose_Hubbard_subspace}
  \ket{\Psi}= \sum_{n=0}^N c_n \ket{n,N-n}\;,
\end{equation}
can be evaluated as we now explain. First, since $H_N$ is invariant under the exchange of the two wells (thanks to the symmetry of $V_N$),
its non-degenerate ground state is invariant under the exchange of $\uHp$ and $\uHm$, i.e., it satisfies
$c_{N-n} = c_n$ for any $n=0,\cdots,N$.

By using (\ref{eq-energy_with_density_matrices}) and the fact that the one-body density matrix $\gamma_\Psi^{(1)}$ has a two-dimensional range spanned by $\uHm$ and $\uHp$, 
the kinetic and potential energies of the state \eqref{eq-general_states_Bose_Hubbard_subspace} reads 
\begin{equation*}
  E_\Psi^{\rm (kin+pot)} = e_+ \langle \uHp, \gamma_\Psi^{(1)} \uHp \rangle
  + e_- \langle \uHm , \gamma_\Psi^{(1)} \uHm \rangle + 2 T_N \re \langle \uHp, \gamma_\Psi^{(1)} \uHm \rangle
\end{equation*}
with $e_\pm = \langle \uHpm , ( -\Delta + V_N ) \uHpm \rangle$.
From  (\ref{eq:other def DM}) and  the identity $c_n=c_{N-n}$
one concludes that  
\begin{equation} \label{eq-diagonal_elements_one-body_density_matrix}
\langle \uHpm, \gamma_\Psi^{(1)} \uHpm \rangle= \bra{\Psi} a_\pm^\ast a_\pm \ket{\Psi}= \frac{N}{2}\;.
\end{equation}
Calculating on the other hand
$$
\langle \uHp, \gamma_\Psi^{(1)} \uHm \rangle= \bra{\Psi} a_+^\ast a_{-} \ket{\Psi},
$$
one gets
\begin{equation} \label{eq-energy_kin_pot_general_Psi}
  E_\Psi^{\rm (kin+pot)} 
  =  E_{\rm loc}^{\rm (kin+pot)}  + 2 T_N \sum_{n=0}^N \sqrt{n ( N-n+1 )} \re \{ \overline{c_n} c_{n-1} \}
\end{equation}
with $E_{\rm loc}^{\rm (kin+pot)} = ( e_+ + e_- ) N /2$.

By arguing as in Section~\ref{sec:heuristic}, discarding all matrix elements of the interaction $w$ save for those
between $\uHpm^{\otimes 2}$ and  $\uHpm^{\otimes 2}$, one obtains from \eqref{eq-interaction_energy_general_Psi}
the interaction energy
\begin{equation}  \label{eq-energy_int_general_Psi}
 E_\Psi^{\rm (int)} 
=
   U_N
     \left(
     \frac{N (N-2)}{4} + \left\langle \left(\Delta \N_{-}\right)^2\right\rangle_\Psi
     \right)
     ,\quad
U_N=\frac{\lambda \langle \uH^{\otimes 2} ,  w \, \uH^{\otimes 2} \rangle}{N-1}      
\end{equation}  
with
$$
\left\langle \left(\Delta \N_{\pm}\right)^2\right\rangle_\Psi = \left\langle  \Psi \left| \left(\N_{\pm} - N/2 \right) ^2 \right| \Psi \right\rangle
$$
the square fluctuation of $\N_{\pm} = a_\pm^\ast a_\pm$. Thus, the total energy of a state of the form~\eqref{eq-general_states_Bose_Hubbard_subspace} is given by
$$ E_\Psi = E_\Psi^{\rm (kin+pot)} + E_\Psi^{\rm (int)} =  E_{\rm loc} + 2T_N \sum_{n=0}^N \sqrt{n ( N-n+1 )} \re \{ \overline{c_n} c_{n-1} \} + U_N \left\langle \left(\Delta \N_{-}\right)^2\right\rangle_\Psi\;,$$
where $ E_{\rm loc}$ is the energy of the localized state, see Section~\ref{sec:heuristic}.
From these considerations, we moreover
deduce that the problem of finding the  state ${\Psi}$ in the subspace $\gH_{\rm BH}$  with minimal energy
$E_\Psi$
is equivalent to determining the ground state of the following two-mode Bose-Hubbard Hamiltonian acting on $\gH_{\rm BH}$
\begin{equation*}
  H_{\rm BH}  = 
  e_+ \N_+ + e_- \N_{-} + T_N ( a_{-}^\ast a_+ + a_+^\ast a_{-} ) +
  \frac{U_N}{2} \left( a_+^\ast a_+^\ast a_{+} a_{+} + a_{-}^\ast a_{-}^\ast a_{-} a_{-} \right)\;.
 \end{equation*}
Note that $\N_{-}+\N_+ = N \1$ (here $\1$ denotes the identity operator)
since we are neglecting 
all Bogoliubov excitations outside the one-particle subspace spanned by $\uHm$ and $\uHp$.

To obtain the transition values of Table~\ref{tab1}, consider a trial state  $\Psi$ given by 
\begin{equation} \label{eq-Gaussian_state}
    c_n = \frac{1}{{\mathcal{Z}}_N} e^{-(n-N/2) ^2 / \sigma_N^2}
\end{equation}
with $\sigma_N$ setting the scale of the particle number fluctuations and ${\mathcal{Z}}_N$ a normalization constant. Assuming squeezed particle number fluctuations, $1 \ll \sigma_N \ll N^{1/2}$, simple calculations and estimates give, to leading order in $N$,
\begin{equation}\label{eq:ener sigma}
E_\Psi^{\rm (kin+pot)} + E_\Psi^{\rm (int)} \approx E_{\rm loc} + T_N N \left(1-\frac{1}{2\sigma_N ^2}\right) + U_N \frac{\sigma_N^2}{4}.
\end{equation}
On the other hand, from the computations of Section~\ref{sec:heuristic} we have
$$ E_{\rm dloc} \approx E_{\rm loc} + T_N N + \frac{\lambda}{4} \langle \uH^{\otimes 2} ,  w \, \uH^{\otimes 2} \rangle\;.$$
To minimize~\eqref{eq:ener sigma} in $\sigma_N$, we pick (recall that $T_N < 0$)
$$  
\sigma_N = \left( \frac{2 |T_N| N}{U_N} \right)^{1/4}  \ll N^{1/2} \;\;\mbox{ if }\;\; |T_N| \ll \lambda 
$$
and obtain, for two fixed numbers $a_1,a_2 >0$
$$ E_\Psi - E_{\rm dloc} = a_1 (\lambda |T_N| )^{1/2} - a_2 \lambda < 0 \;\;\mbox{ if }\;\; |T_N| \ll \lambda$$
and the other way around if $|T_N| \gg \lambda$.
One can similarly
show that the state $\Psi$ has a smaller energy than  $\Psiloc$ when
$|T_N| \gg \lambda N^{-2}$ and the other way around if $|T_N| \ll \lambda N^{-2}$.
This
leads to the transitional values of Table~\ref{tab1}. A more precise guess (spin-squeezed state) can be made for the ground state in the Josephson regime, see Appendix~\ref{app:BH}. 

Further note that, both the trial state above and the spin-squeezed state discussed in Appendix~\ref{app:BH} have tunneling factors close to that of the delocalized state (as indicated in~Table~\ref{tab1}). This implies that their one-body density matrix are close to that of the delocalized state to leading order in $N$, and we expect the same for the true ground state.

\subsection{Organization of the proofs}

The rest of the paper is organized as follows:
\begin{itemize}
 \item Section~\ref{sec:hartree} contains useful estimates on the Hartree minimizers to be used throughout the paper, in particular sharp decay estimates.
 \item Section~\ref{sec:bog} recalls those elements of Bogoliubov's theory we shall need in the proofs of our main results, following mainly~\cite{GreSei-13,LewNamSerSol-13}.
 \item Section~\ref{sec:up bound} is concerned with the construction of a trial state having energy~\eqref{eq:ener lim loc}, thus providing the desired  upper bound on the ground state energy.
 \item In Section~\ref{sec:low bound loc} we present the core of the proof of our main theorem, namely the energy lower bound and the estimates on particle number fluctuations that follow from it.
 \item Appendix~\ref{sec:app} contains, for the convenience of the reader, elements of proofs for the results on Bogoliubov's theory we use in the paper. We make no claim of originality here and refer to~\cite{GreSei-13,LewNamSerSol-13,Seiringer-11} for full details.
 \item We present in Appendix~\ref{sec:app 2} the proof of a lemma used in Section~\ref{sec:low bound loc} about the optimal way of distributing particles between the two wells.
 \item Finally, some details on squeezed states are given in Appendix~\ref{app:BH}.
\end{itemize}

\section{Bounds on the minimizers of the mean-field functionals}\label{sec:hartree}
\subsection{Hartree minimizer in a single well}\label{sec:single well}

The Hartree functional that we shall study is 
\begin{equation}\label{eq:Hartree func 2}
\EH[u] = \int_{\R ^d} \Bigl( |\nabla u | ^2 + V |u| ^2 
+ \frac{\lambda}{2} |u| ^2 w \ast |u| ^2 \Bigr) \;, 
\end{equation}
where $V(x)$ is the single well potential~\eqref{eq:hom pot}.
 We shall denote the minimizer of $\EH[u]$ by
 $\uH$. We will later apply the results of this section 
to $\uHm$ and $\uHp$, that are just translates of $\uH$.

Given Assumption~\ref{assumptions-w} on the interaction $w$, 
the existence and uniqueness of the minimizer of  the Hartree functional \eqref{eq:Hartree func 2}
under the unit mass constraint is an easy exercise.
In fact, since $w$ is assumed to be of positive type, the functional \eqref{eq:Hartree func 2} is strictly convex in $|u| ^2$.
It follows from the identity $| \nabla u |^2 = ( \nabla | u |)^2 + |u|^2 ( \nabla \varphi)^2 $ with $u = |u| e^{\I \varphi}$ that 
$\EH[u] \geq \EH [ |u|]$, with equality if and only if $\varphi$ is constant.
Thus the minimizer $\uH$ is unique up to a constant phase factor, which can be chosen such that 
$\uH >0$. One can also show that $\uH$ is radial (see e.g.~\cite{LieSeiYng-00} for details on these claims).
By exploiting the elliptic character of the variational equation satisfied by $\uH$ (see~\eqref{eq:var eq Hartree}), 
one shows in the usual way that $\uH$ is a smooth function.

\begin{proposition}[\textbf{Decay estimates on the Hartree minimizer}]\label{pro:Hartree}\mbox{}\\
Let $V(x) = |x| ^s$ with $s \geq 2$,  $A(r)$ be the Agmon distance~\eqref{eq:Agmon}, and 
\begin{equation}\label{eq:alpha}
\alpha= \displaystyle \begin{cases}
\frac{2d-2+s}{4s} & \mbox{ if } s>2 \\[1ex]
\frac{2d-2+s}{4s} - \frac{\mu}{2s} & \mbox{ if } s=2\;,         
        \end{cases}
\end{equation}
with $\mu$ the chemical potential in~\eqref{eq:var eq Hartree}. 

For any $0<\eps<1$, and any $|x|\geq R_0$ large enough, $\uH$ satisfies the pointwise estimates
\begin{equation} \label{eq:decay Hartree}
c_\eps \frac{e^{-A(|x|)}}{V(x)^{\alpha+\eps}}  \leq \uH (x) \leq C_{\eps} \frac{e^{-A(|x|)}}{V(x)^{\alpha-\eps}}
\;,
\end{equation} 
where  $c_{\eps}>0$ and $C_{\eps}>0$ are two constants depending only on $\varepsilon$.
\end{proposition}

In the special case of a harmonic trap $V(x)=|x|^2$, 
this shows that $\uH(x)$ decays like a Gaussian when $|x| \to \infty$: 
\begin{equation}\label{eq:decay Hartree prime}
\uH (x) \sim C \exp\left( - \onehalf |x|^{2} \right)
\end{equation}
up to some power-law corrections.

\begin{proof}
We set, for some number $\beta \in \R$,
\begin{equation}\label{eq:sub solution}
f (x) = \exp\left( -  A (|x|) \right) V(x) ^{-\beta/s} = \exp\left( - \left( 1+\frac{s}{2} \right) ^{-1} |x| ^{1+s/2}\right) |x| ^{-\beta}\;. 
\end{equation}
Then, setting $r = |x|$,
\begin{equation}\label{eq:Delta f}
\Delta f (x) = \left[ r^s + \left(2\beta - \frac{s}{2} - d + 1 \right) r ^{s/2 -1} + \left( \beta ^2 + 2 \beta - d \right) r ^{-2} \right] f (x)\;.
\end{equation}
Since $|\uH| ^2$ decays at infinity and $w$ has compact support, $w* |\uH| ^2$ also decays at infinity. We deduce that, for $r$ large enough, 
\begin{equation}\label{eq:sup sub equation}
\left(-\Delta f + V + \lambda w* |\uH| ^2- \mu\right) f  \geq 0 \quad (\mbox{respectively } \leq 0)
\end{equation}
if one picks $\beta = s\alpha -\eps$ (respectively $\beta = s\alpha + \eps$). The result is obtained by using the above functions as super/sub-solutions for the variational equation~\eqref{eq:Hartree eq} and a maximum principle argument. 

Pick first $\beta = s\alpha - \eps$, define $R_-$ to be some radius large enough for~\eqref{eq:sup sub equation} and 
\begin{equation}\label{eq:ball condition}
 r ^s > \mu - \lambda w*|\uH| ^2 (x) 
\end{equation}
to hold whenever $r\geq R_-$. Let $f_-$ be equal to $f$ outside $B(0,R_-)$ and smoothly extended to a function bounded away from $0$ inside $B(0,R_-)$. Further set  
$$ C_\eps = \max_{|x|< R_{-}} \Bigl\{ \frac{\uH (x)}{f_- (x)} \Bigr\} > 0 $$
and 
$$ g_- = \uH - C_\eps f_- \;.$$
The latter being a smooth function, decaying at infinity, it must reach a global maximum. We have the following alternative:
\begin{enumerate}
\item Either $g_- $ reaches its maximum at a point $x_0$ inside $B(0,R_-)$, then by construction 
$$\uH (x) - C_\eps f_- (x)  \leq \uH (x_0) - C_\eps f_{-}(x_0) \leq 0$$
for all $x$.

\item Or $g_- $ reaches its maximum  at some point $x_0$ outside $B(0,R_{-})$. Then, according to~\eqref{eq:sup sub equation} and the variational equation~\eqref{eq:var eq Hartree}, we have
$$
\left(-\Delta + V +\lambda w \ast | \uH|^2 - \mu \right) g_- (x_0) = - C_\eps \left(-\Delta + V  +\lambda w \ast | \uH|^2 - \mu \right) 
f_{-} (x_0) 
\leq 0
$$
But $\Delta g_{-} (x_0)\leq 0$ because $x_0$ is a maximum of $g_{-}$, thus 
$$ \left( |x_0|^s +\lambda w \ast | \uH|^2 (x_0) - \mu \right) g_- (x_0) \leq 0$$
and this implies $g_{-} (x_0) \leq 0$ upon inserting~\eqref{eq:ball condition}. Hence   
$$\uH (x) - C_\eps f_- (x)  \leq \uH (x_0) - C_\eps f_{-}(x_0) \leq 0$$
for all $x$ again.
\end{enumerate}
In both cases one has $g_{-} (x) \leq g(x_0) \leq 0$ for all $x \in \R^d$, which yields the upper bound in~\eqref{eq:decay Hartree} because $f_- (x)= f (x)$ for  $|x|$ large enough.

The lower bound in~\eqref{eq:decay Hartree} is proven similarly, picking now $\beta = s\alpha + \eps$, defining $f_+$ similarly as before and setting
$$c_\eps = \min_{|x| < R_{+}} \left\{ \frac{\uH (x)}{f_{+}(x)}\right\}, \quad g_{+}  = \uH  - c_\eps f_{+}. $$ 
The latter function, being smooth and decaying at infinity, must reach a global minimum or else be everywhere positive. In the latter case there is nothing to prove, while in the former one can argue exactly as above, switching some signs where appropriate.
\end{proof}

The pointwise estimates~\eqref{eq:decay Hartree} yield a simple but useful corollary, namely a control of the mean-field
potential  generated by $|\uH| ^2$ via $w$,
\begin{equation}\label{eq:MF pot}
\hH := w* |\uH| ^2.
\end{equation}

\begin{lemma}[\textbf{Local control of the mean-field potential}]\label{lem:control pot}\mbox{}\\
For any $\eta >0$, there is a constant $C_{\eta }>0$ such that 
\begin{equation}\label{eq:decay pot}
 \hH \leq C_\eta |\uH| ^{2-\eta}. 
\end{equation}
\end{lemma}

Note that if $w$ was a contact potential $w=\delta_0$, \eqref{eq:decay pot} would be an equality with $\eta = 0$, $C_\eta = 1$. What the lemma says is that the decay of the mean-field potential is not much worse than in the case of purely local interactions.

\begin{proof}
Since $w$ is bounded with a compact support included in the ball $B(0,R_w)$, one has
$$ 
\hH (x) = \int_{\R ^d} w(y) |\uH (x-y)| ^2 \D y \leq \| w \|_\infty | B(0,R_w)| \sup_{y\in B(0,R_w)} |\uH (x-y)| ^2 
\;.
$$
Let $\eps \in (0,1)$ be such that
$$\frac{1+\eps}{\left(1-\eps\right)^2} = \left(1 - \frac{\eta}{2}\right)^{-1}.$$
We can use the following estimate on the Hartree minimizer, which is less precise than~\eqref{eq:decay Hartree}: 
\begin{equation} \label{eq:decay Hartreebis}
c_{\eps} \exp\left( - (1 + \eps) A ( |x|)  \right) \leq \uH (x) \leq C_{\eps} 
\exp\left( -  (1 - \eps) A(|x|) \right)\quad  \quad x\in \R ^d\;,
\end{equation}
where  $c_{\eps}>0$ and $C_{\eps}>0$ are two constants depending on $\varepsilon$.
Therefore, we have for any $x \in \R^d$, $|x|\geq R_w$,
\begin{eqnarray} \label{eq-proof_lem:control pot} 
\nonumber
\sup_{y\in B(0,R_w)} \frac{|\uH (x-y)|}{|\uH (x)| ^{1-\eta/2}} 
& \leq & 
C_{\eps}' \sup_{y\in B(0,R_w)} \exp\Big( - (1-\eps)  A(|x-y|)+ (1+\eps)\big(1-\frac{\eta}{2}\big) A ( |x|) \Big)
\\
& \leq &
C_{\eps}' \exp\Big( - (1-\eps) \big( A(|x|-R_w)- (1-\eps) A ( |x|) \big) \Big)
\end{eqnarray}
for some $C_\eps'>0$. It is easy to check that the exponential in the second line of \eqref{eq-proof_lem:control pot} is bounded in $x$, 
which yields the desired result. 
\end{proof}

A further consequence of Proposition~\ref{pro:Hartree} is the estimate on the tunneling energy announced in~\eqref{eq:estim tunnel}. Let $\uHm(x) = \uH (x+\mathbf{x}_N)$ and $\uHp (x-\mathbf{x}_N)$ be the Hartree minimizers corresponding to the left and right 
trapping potentials $V_N^-$ and $V_N^+$, respectively.
%
  
\begin{proposition}[\textbf{Bounds on tunneling terms}]\label{cor:tunneling}\mbox{}\\
For any $\eps >0$, one can find some positive constants $c_\epsilon$, $C_\varepsilon$, and $L_\eps >0$ such that
for any $L_N > L_\eps$, 
\begin{equation}\label{eq:tunn 1}
c_\eps \exp\left( - 2 (1 + \eps) A \Big( \halfLN \Big) \right) \leq \int_{\R ^d } \uHm \uHp \leq 
C_\eps \exp\left( - 2 (1 - \eps) A \Big( \halfLN \Big) \right).
\end{equation}
Moreover, the tunneling energy defined in~\eqref{eq:tunneling term} is negative for $N$ large enough and satisfies 
\begin{equation}\label{eq:tunn 2}
 c_\eps \exp\left( -  2 (1 + \eps) A \Big( \halfLN \Big) \right)  \leq | T_N | \leq 
C_\eps \exp\left( -  2 (1 - \eps) A \Big( \halfLN \Big) \right) \;.
\end{equation}
\end{proposition}

In fact, since we know the rate of decay of the Hartree minimizer down to polynomial corrections, we could reach a similar precision in the estimates~\eqref{eq:tunn 1} and~\eqref{eq:tunn 2}. We do not state this explicitly for conciseness. We will, 
however, need this information to prove that $T_N <0$ and get the lower bound on $|T_N|$ in \eqref{eq:tunn 2}.


\begin{proof}
We use the bounds~\eqref{eq:decay Hartree}, suitably translated by $\pm \mathbf{x}_N$. For the first estimate, the polynomial correction to the rate of decay obtained in~\eqref{eq:decay Hartree} is not relevant and one can just calculate integrals of the form
\begin{equation}\label{eq:integrals}
I_a = \int_{\R^d} e^{ -a |x-\mathbf{x}_N| ^{1+s/2} - a |x+\mathbf{x}_N| ^{1+s/2} } \D x 
\end{equation}
with $a =\left((1+s/2) ^{-1} \pm \eps \right)$, the $\pm \epsilon$ being used to absorb any additional polynomial term.

On the one hand, for $|x| \geq C L_N$ with a large enough constant $C>0$ we have, by the triangle inequality 
$$ |x-\mathbf{x}_N| ^{1+s/2} + |x+\mathbf{x}_N| ^{1+s/2} \geq 2 |x| ^{1+s/2} \left(1 - \frac{1}{2C}\right)^{1+s/2}\;. $$
Thus, provided $C$ is chosen large enough we obtain 
\begin{align*}
\int_{ \{ |x| \geq C L_N\} } e^{ -a |x-\mathbf{x}_N| ^{1+s/2} - a |x+\mathbf{x}_N| ^{1+s/2} } \D x  
&\leq \int_{\{ |x| \geq C L_N\} } e^{-2 a  \left(1 - \frac{1}{2C}\right)^{1+s/2} |x| ^{1+s/2} } \D x  \\
&\leq e^{- a  \left(1 - \frac{1}{2C}\right)^{1+s/2} C ^{1+s/2} L_N ^{1+s/2} } \int_{\{ |x| \geq C L_N\} } 
e^{- a  \left(1 - \frac{1}{2C}\right)^{1+s/2} |x| ^{1+s/2} } \D x  \\
&\leq e^{- 3  \left(1 - \eps \right) A \Big( \halfLN \Big) }
\end{align*}
which is much smaller than the precision we aim at in the desired result.

There remains to estimate the part of the integral located where $|x|\leq C L_N$. We write $x= (x_1,\ldots,x_d)$ and note that the function 
\begin{multline*}
 |x-\mathbf{x}_N| ^{1+s/2} + |x+\mathbf{x}_N| ^{1+s/2}  = \left( \left|x_1-\frac{L_N}{2}\right| ^2 + |x_2| ^2 + \ldots + |x_d| ^2  \right) ^{1/2 + s/4} 
 \\ + \left( \left|x_1+\frac{L_N}{2}\right| ^2 + |x_2| ^2 + \ldots + |x_d| ^2  \right) ^{1/2 + s/4} 
\end{multline*}
is even and convex in $x_1$. It thus takes its absolute minimum at $x_1 = 0$ and we have 
\begin{align*}
 |x-\mathbf{x}_N| ^{1+s/2} + |x+\mathbf{x}_N| ^{1+s/2}  &\geq 2 \left( \frac{L_N}{2} \right) ^{1 + s/2}  \left( 1 + \frac{4|x_2| ^2}{L_N^2} + \ldots + \frac{4|x_d| ^2}{L_N^2}  \right) ^{1/2 + s/4}\\
 &\geq 2 \left( \frac{L_N}{2} \right) ^{1 + s/2} + 2 \left( |x_2| ^{1+s/2} + \ldots + |x _d| ^{1+s/2} \right)
\end{align*}
and it follows that 
$$\int_{ \{ |x| \leq C L_N\} } e^{ -a |x-\mathbf{x}_N| ^{1+s/2} - a |x+\mathbf{x}_N| ^{1+s/2} } \D x \leq C L_N  e^{- 2 \left(1 + \eps \right) A \Big( \halfLN \Big) }\;, $$
where we separate the integrals in $x_2, \ldots,x_d$, which are all convergent. The prefactor $L_N$ comes from the integral in $x_1$ and can be absorbed in the exponential, changing slightly the value of $\eps$, to obtain the upper bound in~\eqref{eq:tunn 1}. For the lower bound we simply note that 
$$ \int_{\R^d} e^{ -a |x-\mathbf{x}_N| ^{1+s/2} - a |x+\mathbf{x}_N| ^{1+s/2} } \D x  \geq \int_{\{ |x| \leq L_N ^{\gamma} \} } 
e^{ -a |x-\mathbf{x}_N| ^{1+s/2} - a |x+\mathbf{x}_N| ^{1+s/2} } \D x $$
for any $\gamma$. In particular, taking $\gamma <1$, we have by a Taylor expansion 
$$ |x-\mathbf{x}_N| ^{1+s/2} + |x+\mathbf{x}_N| ^{1+s/2} = 2  \left( \frac{L_N}{2} \right) ^{1 + s/2} + O \left( L_N ^{\gamma + s/2} \right) \leq 2  \left( \frac{L_N}{2} \right) ^{1 + s/2} \left( 1+o(1) \right) $$
on the relevant integration domain, thus 
$$ \int_{\R^d} e^{ -a |x-\mathbf{x}_N| ^{1+s/2} - a |x+\mathbf{x}_N| ^{1+s/2} } \D x \geq e^{- 2  \left(1 - \eps \right) A \Big( \halfLN \Big) } L_N ^{d\gamma} $$
and we can again absorb the last factor $L_N ^{d\gamma}$ in the exponential, changing slightly $\eps$.

For the second estimate~\eqref{eq:tunn 2}, we use the expression~\eqref{eq:tunneling termbis}: 
$$ T_N = \int_{\R ^d } \uHm \Vtp \uHp = \int_{\R ^d} \uHm \left( V_N - V_N^+\right)  \uHp + \int_{\R ^d} \uHm \left( \mu - \lambda w * |\uHp| ^2 \right)  \uHp.$$
To get the bounds~\eqref{eq:tunn 2} one can estimate exactly as above, absorbing any polynomial growth coming from $V_N - V_N^+$ into exponential factors. To prove that $T_N <0$, a little more care is needed, and we use the full information contained in~\eqref{eq:decay Hartree}, namely that we know the rate of decay up to polynomial corrections. 

Estimating as previously, keeping track of polynomial factors, we obtain, for any $\eps >0$,
\begin{equation}\label{eq:merde au tunneling}
  \begin{array}{lcl}
    c_{\eps} L_N ^{1 - 2\alpha - \eps} \exp\left( -  2 A \Big( \halfLN \Big) \right) & \leq & \left| \int_{\R ^d} \uHm \left( \mu - \lambda w * |\uHp| ^2 \right)  \uHp \right| \\
    & & \hspace*{2cm} \leq C_{\eps} L_N ^{1-2\alpha + \eps} \exp\left( -  2 A \Big( \halfLN \Big) \right) 
  \end{array}
\end{equation}
where $\alpha$ is defined in~\eqref{eq:alpha}. On the other hand, since $V_N -V_N ^+$ is negative by definition we have for any $\gamma$
$$ \int_{\R ^d} \uHm \left( V_N - V_N^+\right)  \uHp \leq \int_{ \{ |x| \leq L_N ^{\gamma}\} } \uHm \left( V_N - V_N^+\right)  \uHp.$$
Taking $\gamma <1$ we have, on the latter integration domain,
$$ V_N - V_N^+  = 4 s \1_{\{x_1<0 \}} x_1 \left( \frac{L_N}{2}\right) ^{s/2 - 1} \left( 1 + o (1))\right)$$
using a Taylor expansion. Putting this together with~\eqref{eq:decay Hartree} and arguing as above we deduce that 
$$  \int_{\R ^d} \uHm \left( V_N - V_N^+\right)  \uHp \leq - C_\eps L_N ^{1+s/2 - 2 \alpha - \eps} \exp\left( -  2 A \Big( \halfLN \Big) \right).$$
Comparing with~\eqref{eq:merde au tunneling} we see that $T_N <0$ for $N$ large enough as we claimed and one deduces the lower bound on $|T_N|$.
\end{proof}

\begin{remark}[Bounds on the off-site interaction energies]\label{rem:bound_int_energy_term}\mbox{}\\
One can show in a similar way that the off-site terms appearing in the interaction energy of the 
localized and delocalized states in Sec.~\ref{sec:heuristic},  
$$
\langle \uHm \otimes \uHp , w\,  \uHm \otimes \uHp\rangle
\;,\; 
\langle \uHm^{\otimes 2} , w\,  \uHp^{\otimes 2} \rangle
\;,\;
\langle \uHpm^{\otimes 2} , w\,  \uHm \otimes \uHp \rangle
$$
are of the order of $\exp ( -  2 (1 - \eps) A ( L_N/2) )$
in the limit $L_N \to \infty$
(this follows immediately from \eqref{eq:tunn 1} and $\| w \|_\infty <\infty$ for the last two terms, and comes from the fact 
that $w$ has compact
support for the first term).\hfill$\diamond$
\end{remark}

\subsection{Hartree energy and minimizer in a perturbed well}\label{sec:pert well}

In the sequel, we shall be lead to consider a perturbation of the previous Hartree functional. This comes about when estimating tunneling effects in energy lower bounds. Essentially, we perturb the functional by a relatively small potential in a region far away form the bottom of the well, and we prove that this  does not change much the Hartree energy and minimizer. 

Consider 
\begin{equation}\label{eq:Hartree func pert}
\EHt [u] = \int_{\R ^d} \Big( |\nabla u | ^2 + \tV (x) |u| ^2 +\frac{\lambda}{2} |u| ^2 w \ast  |u| ^2 \Big)\;. 
\end{equation}
where the perturbed potential is of the form
\begin{equation}\label{eq:pert pot}
\tV (x) = V(x) ( 1 - \delta_N (x)) 
\quad , \quad | \delta_N (x) | \leq \delta \,\1_{\left|x^1 - \halfLN \right| \leq  \ell} (x)
\end{equation}
for some constant $\delta >0$ independent of $N$. Here, $x^1$ is the first coordinate of  $x$, and we thus perturb the original potential in a strip of 
width $\ell$ centered at a distance $L_N/2$ from the origin. The choice of $\ell$ will be discussed later, the point being that if $\ell \ll L_N$ we do not perturb the problem much.

We denote by $\uHt$ and $\eHt$ respectively the unique positive minimizer and the minimum of the above functional. 
They satisfy properties very similar to those of the unperturbed analogues. In particular, one has the same estimates 
on $\uHt$ as in~\eqref{eq:decay Hartreebis} in terms of the Agmon distance associated to the {\emph{unperturbed}} potential
 $V(r)$:

\begin{lemma}[\textbf{Pointwise estimates for the perturbed minimizer}]\label{lem:pert decay}\mbox{}\\
For any $0< \eps <1$, one can find some constant $C_\eps$, $c_\eps$, and $L_\eps >0$ such that for all $L_N > L_\eps$ and 
$x\in \R ^d$,
\begin{equation}\label{eq:decay pert}
c_{\eps} \exp\left( - (1 + \eps) A ( |x|)  \right) \leq \uHt (x) \leq C_{\eps} 
\exp\left( -  (1 - \eps) A(|x|) \right)
\end{equation}
\end{lemma}

\begin{proof}
One only needs minor modifications to the arguments in the proof of Proposition~\ref{pro:Hartree}.
\end{proof}

\medskip

We next prove that the difference between the trapping energies in the perturbed and unperturbed potentials $V$ and $\tV$,
$$
\int_{\R^d} V \delta_N | \uH|^2 
$$
is of the order of the tunneling energy to some power arbitrary close to one. 

\begin{lemma}[\textbf{Difference in the trapping energies}]\label{lem:pert enr}\mbox{}\\
For any $0< \eta <1$, one can find some constant $C_\eta$, $c_\eta$, and $L_\eta >0$ such that for all $L_N > L_\eta$,
\begin{equation}\label{eq:difference trap energies}
\int_{\R^d} V |\delta_N|  |\uH|^2  
 \leq  C_{\eta}  |T_N|^{1-\eta}
\quad , \quad  
\int_{\R^d} V  |\delta_N| | \uHt |^2   
 \leq  C_{\eta}  |T_N|^{1-\eta}\;.
\end{equation}
\end{lemma}

\begin{proof}
Using  \eqref{eq:pert pot},
bounding $V(r)$ by $e^{\eps A (r)}$ for large $r$  and  using
the pointwise estimates~\eqref{eq:decay Hartree}, we have for large enough  $L_N$ 
\begin{eqnarray*}
& & \int_{\R^d} V |\delta_N|  \uH^2  
\leq  \delta   \int_{\{ |x^1- \halfLN | \leq \ell\} } V | \uH |^2   
\leq  
C_\eps^2 \delta   \int_{\{ |x^1- \halfLN | \leq \ell\} } \exp \Big( -2 ( 1 - \eps) A( |x| ) \Big) \D x
\\
& & \hspace*{1cm} \leq  
2 C_\eps^2 \delta   \ell e^{-2 ( 1 - \eps) A(L_N/2-\ell)}
 \int_{\R^{d-1}}
\exp \left( -2 ( 1 - \eps) \int_{\halfLN-\ell}^{\left( \left( L_N/2 - \ell\right)^2 + | x^\bot|^2\right) ^{1/2}} \sqrt{V(r)} \D r \right) \D x^\bot
\;.
\end{eqnarray*}
Arguing similarly as in the proof of Proposition~\ref{cor:tunneling}, the last integral can be bounded from
above by a polynomial function of $\halfLN - l$ and thus by $e^{\eps A (L_N/2)}$ for large enough $L_N$. We also have
$$
A \Big( \halfLN \Big) - A \Big( \halfLN - \ell \Big) \leq \ell \sqrt{V\Big( \halfLN \Big)} 
\leq \frac{\eps}{2(1-\eps)} A \Big( \halfLN \Big)\;.
$$
Collecting the above bounds and using the lower bound on the tunneling energy in \eqref{eq:tunn 2},
we get the desired result. One proceeds similarly for the proof of the second inequality in \eqref{eq:difference trap energies}, by
relying on Lemma~\ref{lem:pert decay}.  
\end{proof}

We can now prove the announced result that $\uHt$ is close to $\uH$ for large $L_N$. 
We could probably prove stronger estimates, but we refrain from doing so for shortness.

\begin{proposition}[\textbf{Difference between perturbed and unperturbed minimizers}]\label{pro:comp Hartree min}\mbox{}\\
For any $\eta >0$, one can find constants $C_\eta$, $c_\eta$, and $L_\eta>0$ such that if $L_N \geq L_\eta$,  then
\begin{equation}\label{eq:dif ener pert}
 |\eH - \eHt | \leq C_\eta |T_N| ^{1-\eta}
\end{equation}
and
\begin{equation}\label{eq:comp Hartree min}
\norm{\uH - \uHt}_{L^2 (\R ^d)}  \leq C_\eta |T_N| ^{1/2-\eta}. 
\end{equation}
\end{proposition}

\begin{proof}
We proceed in several steps:

\medskip

\noindent\textbf{Step 1.} We clearly have
$$ \eHt = \EHt \big[ \uHt \big] = \EH \big[ \uHt \big] + \int_{\R^d} V \delta_N | \uHt|^2
\geq \eH - \int_{\R ^d} V |\delta_N | |\uHt| ^2
$$
and similarly 
$$ \eH \geq \eHt - \int_{\R ^d} V |\delta_N|  |\uH| ^2\;.
$$
The bound \eqref{eq:dif ener pert} follows from these
inequalities and  from Lemma~\ref{lem:pert decay}.
Thus $\eHt \to \eH$ when $L_N \to \infty$. 
Let  $\eHt^{\eta}$ be the minimum of the energy functional \eqref{eq:Hartree func pert}
in which the potential $\tV$ is replaced by a perturbed potential $\tV + \eta W$, where
$W \in L^\infty (\real^d)$ is a bounded potential. 
By the same argument as above, $\eHt^{\eta}$
converges to $\eH$ when $(L_N^{-1},\eta) \to (0,0)$. 
But $\eHt^{\eta}$ is a concave function of $\eta$ (as infimum of an affine functional), so we deduce that 
its derivative with respect to $\eta$ converges  to the corresponding derivative
for $\delta_N =0$ when  $L_N \to \infty$. These derivatives are given by
$$\frac{\partial \eHt^{\eta} }{\partial \eta} \bigg|_{\eta=0} =  \int_{\R^d} W |\uHt| ^2 \quad \to \quad  
\frac{\partial \eH^{\eta} }{\partial \eta} \bigg|_{\eta=0}= \int_{\R^d} W |\uH| ^2\;.
$$
This being so for any $W \in L^\infty (\real^d)$, we deduce that 
$$ |\uHt| ^2 \to |\uH| ^2 \quad , \quad L_N \to \infty$$
in the $L^{\infty}$ weak-$*$ topology. In particular, since by Assumption~\ref{assumptions-w} the interaction potential $w$ is 
bounded, 
\begin{equation}\label{eq:MF pot conv}
 w * |\uHt| ^2 (x) \to w * |\uH| ^2 (x) 
\end{equation}
for almost all $x \in \R^d$. Furthermore, by the dominated convergence theorem,
\begin{equation} \label{eq-proof convergence chem pot}
\int_{\R^d} \big| \uHt \big|^2  w \ast \big| \uHt \big|^2 
 =  \int_{\R^d} \hat{w} (k) \bigg| \int_{\R^d} e^{\I k x} \big| \uHt \big|^2 \bigg|^2
 \to 
\int_{\R^d} | \uH |^2  w \ast  | \uH |^2 \;.
\end{equation}
As a result, the chemical potentials $\muHt$ and $\mu$ associated to $\uHt$ and $\uH$, respectively,
satisfy (see \eqref{eq-chem pot})
\begin{equation}\label{eq:dif chem pot}
\muHt  \to \mu \quad \text{when} \quad L_N \to \infty\;.
\end{equation}

\medskip

\noindent\textbf{Step 2.} The Hartree minimizer $\uH$ is an eigenfunction with zero eigenvalue
of the mean-field Hamiltonian
$$ 
\HH = -\Delta + V + \lambda w * |\uH| ^2  -\mu \;.
$$
Since it is positive, $\uH$ must be in fact the ground state of this Hamiltonian, which is non degenerate. 
Similarly, $\uHt$ is the non-degenerate ground state of the perturbed Hamiltonian
$$ 
\HHt = -\Delta + \tV  + \lambda w * |\uHt| ^2 -\muHt\;.
$$
By using Lemma~\ref{lem:pert decay} and \eqref{eq-proof convergence chem pot}, we have 
$$ 
\muHt =\big\langle \uHt , ( \HHt + \muHt )  \uHt \big\rangle 
=  \big\langle \uHt ,  ( \HH + \mu ) \uHt \big\rangle + o (1) 
$$
and since $\HH$ has a non-degenerate ground state we deduce 
$$ 
\muHt
\geq \mu + c  \| P^{\perp} \uHt \| ^2 + o (1)\;,
$$
where $c>0$ is the spectral gap of $\HH$ and $P^{\perp}$ the orthogonal projector onto $\uH$. 
One concludes from this inequality and from \eqref{eq:dif chem pot} that 
\begin{equation}\label{eq:MF min conv}
\big\| \uHt - \uH \big\|_{L^2 (\R^d)} \to 0
\quad \text{when} \quad L_N \to \infty\;. 
\end{equation}

\medskip

\noindent\textbf{Step 3.} It follows from  Assumption~\ref{assumptions-w} on $w$ that the Hessian of $\EH$ at $\uH$ is 
non degenerate  (see the related discussions in~\cite[Section 2]{LewNamSerSol-13}). 
Since we already know~\eqref{eq:MF min conv}, for $L_N$ sufficiently large so that $\| \uH - \uHt\|_{L^2 (\R ^d)}$ is small
enough, one has
$$
\EH [\uHt] \geq \EH [ \uH]  + a \norm{\uH - \uHt}_{L^2 (\R ^d)} ^2
$$
for some fixed constant $a>0$. Hence
\begin{multline*}
\eHt - \int_{\R ^d}V \delta_N |\uHt| ^2 = 
 \EH [ \uHt] \geq \EHt [ \uH ] - \int_{\R ^d}V \delta_N |\uH | ^2+ a \norm{\uH - \uHt}_{L^2 (\R ^d)} ^2
 \\ 
\geq  \eHt  - \int_{\R ^d} V \delta_N  |\uH| ^2 + a \norm{\uH - \uHt}_{L^2 (\R ^d)} ^2 
\;.  
\end{multline*}
Hence, one infers from Lemma~\ref{lem:pert enr} 
that for any $0< \eta <1$ and $L_N$ large enough,
\begin{align*}
 \norm{\uH - \uHt}_{L^2 (\R ^d)} ^2 \leq a^{-1} 
\int_{\R ^d} V \delta_N \left( |\uH| ^2 - |\uHt | ^2 \right) 
\leq 2 a^{-1} C_\eta  T_N ^{1-2 \eta}\;
\end{align*}
which is the desired result.
\end{proof}

\section{Elements of Bogoliubov theory}\label{sec:bog}

Here we recall elements of Bogoliubov's theory that are needed in the rest of the paper, following mainly~\cite{LewNamSerSol-13,Nam-thesis}. See also~\cite{CorDerZin-09,DerNap-13,Seiringer-11,GreSei-13,ZagBru-01} for other recent discussions.

\subsection{Bogoliubov Hamiltonian}\label{sec:bog def}

For clarity we first recall how the Bogoliubov Hamiltonian is constructed.

\subsubsection*{Second quantized formalism} Bogoliubov's approximation for the spectrum of a large bosonic system is usually described in a grand-canonical setting where the particle number is not fixed. This means that the Hamiltonian~\eqref{eq:hamil depart}
is extended to the Fock space 
\begin{equation}\label{eq:Fock space}
\cF = \cF (\gH) := \C \oplus \gH \oplus \gH ^2  \oplus \ldots \oplus \gH ^N \oplus \ldots
\end{equation}
in the usual way
\begin{equation}\label{eq:second quantized}
\bH = 0 \oplus H_1 \oplus H_2 \oplus \ldots \oplus H_N \oplus \ldots
\end{equation}
with (note the value of the coupling constant)
$$ H_M 
= \sum_{j=1} ^M \left( -\Delta_j + V (x_j) \right) + \frac{\lambda}{N-1} \sum_{1\leq i < j \leq M} w(x_i-x_j)\;.$$
It is convenient to express this Hamiltonian by using standard bosonic annihilation and creation operators.
 We denote by $u_i$, $i=0,1\ldots$, the vectors of an orthonormal basis of $\gH = L ^2 (\R ^d)$ with $u_0 = \uH$ the Hartree ground state corresponding to 
$H_N$, i.e. the minimizer of the functional~\eqref{eq:Hartree func}. Let $a^*_i = a^* (u_i)$ and $a_i= a_i (u_i)$ be 
respectively the annihilation and creation operators in the mode $u_i$, defined by
\begin{align}
(a^\ast_i \Psi ) (x_1,\ldots , x_{M+1} )& = ( u_i \otimes_{\mathrm{sym}} \Psi )(x_1,\ldots , x_{M+1} ) \in \gH^{M+1} \nonumber\\ 
(a_i \Psi )(x_1,\ldots , x_{M-1}) & =\sqrt{M} \int_{\R^d} \overline{u}_i(x) \Psi (x,x_1,\ldots, x_{M-1}) \D x \in \gH^{M-1}\label{eq:anni cre}
\end{align}
for any $\Psi\in \gH^M$, where the symmetrized tensor product is defined in \eqref{eq-symm_tensor_prod}. Then we have
\begin{equation}\label{eq:second quantized 2}
\bH - \muH \sum_{i=0} ^{\infty}a_i ^* a_i
= \sum_{i,j=0} ^{\infty} \big( h_{ij} -\delta_{ij} \muH )  a_i ^* a_j  + \frac{\lambda}{2(N-1)} \sum_{i,j,k,l=0} ^{\infty} w_{ijkl} \: a_i ^* a_j ^* a_k  a_l 
\end{equation}
with $\delta_{ij}=1$ if $i=j$ and $0$ otherwise and
\begin{align*}
h_{ij}&=\left\langle u_i |  -\Delta +V |  u_j \right\rangle\\ 
w_{ijkl}&=\left\langle u_i\otimes u_j | w | u_k \otimes u_l \right\rangle = w_{ji lk} = \overline{w}_{klij}\;.
\end{align*}
%

\subsubsection*{Bogoliubov's Hamiltonian} 
Bogoliubov's approximation consists in replacing $a^*_0$ and $a_0$ by $\sqrt{N}$ in the expression 
in the \RHS of \eqref{eq:second quantized 2} and then dropping all terms that are more than quadratic 
in the operators $a^*_i,a_i$, $i=1,2,\ldots$ As explained in~\cite{LewNamSerSol-13}, this amounts to 
second-quantizing the Hessian at $\uH$ of the Hartree functional~\eqref{eq:Hartree func}. 
Non-degeneracy of this Hessian is required and the assumption $\hat{w}\geq 0$ is a convenient way of ensuring this.

Removing a constant (coming from terms involving only $a_0 ^*, a_0$)
one ends up with the quadratic Hamiltonian
\begin{equation}\label{eq:bog hamil}
 \bHb := \sum_{i,j= 1} ^{\infty} \left(h_{ij} - \muH \delta_{ij} \right) \: a_i ^* a_j + \frac{\lambda}{2}  \sum_{i,j = 1} ^{\infty} 
\Bigl( w_{00ij} a_i a_j + \overline{w}_{00ij} a_i ^* a_j ^* + 
2 ( w_{0i0j} + w_{0ij0} ) a_i ^* a_j\Bigr)\;.
\end{equation}
Here, we  have used the variational equation satisfied by $\uH$, 
\begin{equation}\label{eq:Hartree eq}
(-\Delta + V ) \uH + \lambda (w\ast|\uH|^2) \uH = \muH \uH\;,
\end{equation}
to discard the linear terms in $a_i ^*$ and $a_i$ and we have neglected terms of the order of $1/\sqrt{N}$.

The above Hamiltonian acts on the Fock space of elementary excitations, namely the Fock space 
$$ \cF_\perp  = \cF (\gHp) := \C \oplus \gH_\perp \oplus \ldots \oplus \left( \gH_{\perp}\right) ^N \oplus \ldots $$
associated to the Hilbert space 
\begin{equation}\label{eq:excited space}
\gHp = \left\{ \uH \right\} ^{\perp}\;. 
\end{equation}
We denote by $\eB$ the lowest eigenvalue of $\bHb$ and write the associated eigenstate as  
\begin{equation}\label{eq:bog ground state}
 \PhiB = \phiB_0 \oplus \phiB_1 \oplus \ldots \oplus \phiB_n \oplus \ldots 
\end{equation}
with $\phiB_n \in \left(\gH_\perp \right)^n$. It is well-known that $\PhiB$ is a quasi-free state, i.e., is entirely characterized 
via Wick's theorem in terms of its generalized one-body density matrix. Given a state $\Gamma$ on $\gH^N_\perp$,
the latter is an operator combining the usual one-body density matrix
$$ \langle u, \gamma^{(1)}_{\Gamma} v \rangle = \tr [a^*(v) a(u) \Gamma] 
$$
with the pairing density matrix defined by 
$$ \langle u, \alpha_{\Gamma} J v \rangle = \tr [a (v) a(u) \Gamma]\;,$$
where $J$ is the complex conjugation and $u,v$ are arbitrary vectors in $\gH$. 
Note that $\gamma^{(1)}_{\Gamma}$ is a (self-adjoint) non-negative operator on $\gH_\perp$, 
whereas $\alpha_\Gamma$ should be interpreted as an operator from $J \gH_\perp$ to $\gH_\perp$ satisfying 
$\alpha_\Gamma ^* = J \alpha_\Gamma J$. 

\begin{remark}[Bogoliubov's energy at small coupling]\label{rem:bog low filling}\mbox{}\\
One can show  that
\begin{equation}\label{eq:estim bog}
- C \lambda ^2\ \leq  \eB \leq 0 
\end{equation}
for some $\lambda$-independent constant $C>0$.
The upper bound follows by simply taking the vacuum as a trial state. A sketch of the proof of the
 lower bound is given in Appendix~\ref{sec:app}.
Hence the Bogoliubov energy $\eB$ goes rapidly to $0$ when $\lambda \to 0$ and 
may be safely dropped when $\lambda$ is small.

\hfill$\diamond$
 
\end{remark}

\subsection{Useful results}\label{sec:bog known}

The proofs of our main results rely heavily on recent results of~\cite{GreSei-13,LewNamSerSol-13,Seiringer-11}
on the Bogoliubov fluctuations in the case of a single well potential $V_N \equiv V$. 
We summarize them here, and give 
 for completeness some elements of proof in~Appendix~\ref{sec:app}.

In this paper, the main application of Bogoliubov's theory will be to provide a control of quantum fluctuations out of the condensate. As in~\cite{LewNamSerSol-13}, we write any $N$-body wave 
function $\Psi_N \in \gH^N$ as 
\begin{equation}\label{eq:unitary 1}
\Psi_N = \sum_{j=0} ^N \uH^{\otimes (N-j)} \otimes_\sym  \varphi_j 
\end{equation}
with $\varphi_j \in \left(\gH_\perp \right)^j$. The convention here is that $\varphi_0$ is simply a number. 
Then one can define the unitary map
\begin{equation}\label{eq:unitary 2}
U_N : \Psi_N \mapsto \varphi_0 \oplus \ldots \oplus \varphi_N 
\end{equation}
which sends $\gH^N $ to the truncated Fock space $\cF ^{\leq N}_\perp$.
Let  $H_N$ be the $N$-body Hamiltonian~\eqref{eq:hamil depart} with a potential $V_N\equiv V$ independent of $N$.
We denote by
\begin{equation}\label{eq:hamil fock}
\bH_N  := U_N \left( H_N - N \eH \left( \lambda \right) \right) U_N ^*
\end{equation}
the corresponding Hamiltonian on $\cF^{\leq N}_\perp$ after subtraction of the mean field contribution.
The following results show that $\bH_N $ 
is closely related to the Bogoliubov Hamiltonian~\eqref{eq:bog hamil} in the limit $N\to \infty$. We denote by $\dGamma (h^\perp)$ 
the  second quantization of the one-body Hamiltonian $h^\perp = P^\perp (-\Delta +V) P^\perp$ acting on $\gH_\perp$, where
$P^\perp$ is the orthogonal projector onto $\{ \uH \}^\perp$. Let 
$$\Npe = \N - a^\ast ( \uH) a ( \uH )$$
be the particle number operator in $\gH_\perp$, with $\N$ the total particle number operator. 


\begin{proposition}[\textbf{Control of fluctuations out of the condensate}]\label{pro:bog low}\mbox{}\\
Let $\eB <0$ be the lowest eigenvalue of $\bHb $. For $N$ large enough, there is a constant $C>0$ such that, as operators on $\cF^{\leq N}_\perp$,
\begin{equation}\label{eq:bog low pre}
\bH_N  \geq C \left(\dGamma (  h^\perp  ) - C\right) 
\end{equation}
and
\begin{equation}\label{eq:bog low}
\bH_N  \geq \eB + N ^{-1}  \big( \Npe  \big)^2 - C N ^{-2/5}\;.
\end{equation}
\end{proposition}

A lower bound on the first eigenvalue $E_N$ of $H_N$ follows easily from~\eqref{eq:hamil fock} and~\eqref{eq:bog low}. A matching upper bound can be obtained by using a trial state:

\begin{proposition}[\textbf{Upper bound for the single-well many-body energy}]\label{pro:GSE single well}\mbox{}\\
Let $E_N$ be the ground state energy of~\eqref{eq:hamil depart} with $V_N \equiv V$
 and consider the $N$-body wave 
function
$$ \Psi_N := c_M \sum_{j=0} ^M \uH ^{\otimes (N-j)} \otimes_\sym \phiB_{j}\;,$$
where $\phiB_j$   is the component in $(\gH_\perp)^j$ of the ground state of $\bHb$
(cf.~\eqref{eq:bog ground state}), 
$c_M$ a normalization constant, and $M\propto N^{1/5}$ when $N\to \infty$. Then for any $\eps>0$, there is a constant
$C_\eps>0$ such that
\begin{equation}\label{eq:ener exp single well}
E_N \leq \left\langle \Psi_N |  H_N | \Psi_N \right\rangle \leq N \eH + \eB + C_\eps N ^{-2/5 + \eps}\;. 
\end{equation}
\end{proposition}

We conclude this section with a mild decay estimate for the Bogoliubov ground state. 

\begin{lemma}[\textbf{Decay of the Bogoliubov ground state}]\label{lem:decay bog}\mbox{}\\
Let $\PhiB$ be the ground state of $\bHb$ on $\cF_\perp$. 
Then $\gamma_{\PhiB} ^{(1)}$ is trace-class and $\alpha_{\PhiB}$ is Hilbert-Schmidt. 
Let $\rho_{\PhiB} (x) = \gamma ^{(1)}_{\PhiB} (x,x)$ be the one-body density of $\PhiB$. 
Then there is a constant $C>0$ such that 
\begin{equation} \label{eq-tr_pot_is_bounded}
\int_{\R ^d} V \rho_{\PhiB} \leq C\;.
\end{equation}
\end{lemma}

Note that, according to the first statement, the mean number of particles which are not condensed,
$$
\tr \big[ \gamma_{\PhiB} ^{(1)} \big] = \int_{\R^d} \rho_{\PhiB} (x) \D x = 
 \sum_{j=1}^\infty j \big\| \phiB_j \big\|^2\;,
$$
is finite in the limit $N \to \infty$.

A sketch of the proofs of the last three results, following mostly~\cite{LewNamSerSol-13}, is provided in 
Appendix~\ref{sec:app} for the convenience of the reader.

\section{Energy upper bound}\label{sec:up bound}

To prove the energy upper bound in the localized regime we use a trial state where exactly half of the particles is localized 
in each well: 
\begin{equation}\label{eq:trial loc}
\Psi_{\rm loc} := \Psi_-' \otimes \Psi_+ ' \quad , \quad \Psi_- ',\Psi_+ ' \in \gH ^{N/2}\;, 
\end{equation}
where the states $\Psi_-'$ and $\Psi_+'$ describe Bose condensates with $N'=N/2$ particles localized in the $V_N ^-$ and $V_N^+$ wells, 
respectively, together with their Bogoliubov fluctuations. We define them as follows: let 
$$ H_{N'} := \sum_{j=1} ^{N'} \left( -\Delta_j + V(x_j) \right) 
+ \frac{\lambda'}{N'-1} \sum_{1\leq i < j \leq N'} w(x_i-x_j) 
$$
be the Hamiltonian associated  to $N'=N/2$ particles in the single-well potential $V$ 
with a renormalized $N$-dependent coupling constant
$\lambda'$ such that 
$$\frac{\lambda '}{N'-1} = \frac{\lambda}{N-1},$$
and let
$$
\EHN [u] = \int_{\R ^d} \Big( |\nabla u | ^2 + V |u| ^2 
+ \frac{\lambda'}{2}  |u| ^2 w \ast |u | ^2  \Big)
$$
be the corresponding energy functional. 
We denote by $\uHN$ the minimizer of $\EHN [u]$ with unit $L^2$-norm and by
$$\PhiBN = \phiBN _0 \oplus \ldots \oplus \phiBN _j \oplus \ldots \in \cF (\lbrace \uHN \rbrace ^{\perp})$$
the ground state of the Bogoliubov Hamiltonian obtained from $H_{N'}$
by the procedure described in the previous section. 
We then define, as in Proposition~\ref{pro:GSE single well}, the normalized wave-function
\begin{equation}
\Psi '  := c_{N'} \sum_{j=0} ^{M} (\uHN ) ^{\otimes (N'-j)} \otimes_\sym \phiBN_{j} \in \gH^{N'}
\end{equation}
with $M\propto N ^{1/5}$ and $c_{N'}$ a normalization factor. Then, let
\begin{equation}\label{eq:trial loc -}
\Psi_-' := \Psi' (.-\mathbf{x}_N)\quad , \quad \Psi_+'  = \Psi ' (.+\mathbf{x}_N)
\end{equation}
where the translation by $\pm \mathbf{x}_N$ is understood to act on each of the coordinate vectors $x_j$, $j=1,\ldots,N'$, in the argument of the function. 
Similarly, we denote by  $\PhiBN_-$ and $\PhiBN_+$ the translates of the Bogoliubov ground state $\PhiBN$.

We shall prove the following, which gives the desired energy upper bound on the ground state energy $E(N)$ 
 of the $N$-body Hamiltonian~\eqref{eq:hamil depart} in the double well:

\begin{proposition}[\textbf{Energy of the localized state}]\label{pro:loc up bound}\mbox{}\\
Let $\Psiloc$ be the trial state defined in~\eqref{eq:trial loc}. In the localized regime~\eqref{eq:loc regime} we have 
\begin{equation}\label{eq:upper bound}
\frac{E(N)}{N} \leq N ^{-1} \langle \Psiloc | H_N | \Psiloc \rangle \leq 
\eH \left( \Delta_N \frac{\lambda}{2} \right) + \frac{2}{N} \eB \left( \frac{\lambda}{2} \right) + o(|T_N|) + o(N^{-1})
\end{equation}
where $\Delta_N$ and $T_N$ are given by~\eqref{eq:delta N} and~\eqref{eq:tunneling term}.
\end{proposition}

\begin{proof}
Note that~\eqref{eq:trial loc} is not fully symmetric under
 particle exchange, only $\Psi_-'$ and $\Psi_+'$ are.  We thus start with:
 
\medskip

\noindent\textbf{Step 1:~\eqref{eq:trial loc} is an admissible trial state.} It is well-known (see e.g.~\cite[Section~3.2]{LieSei-09}) that 
the ground state energy of the $N$-body Hamiltonian~\eqref{eq:hamil depart}
acting on the unsymmetrized Hilbert space $\gH^{\otimes N}$ coincides with the bosonic ground state energy $E (N)$.  Note that $\Psiloc$ is normalized since $\Psi_- '$ and $\Psi_+ '$ are. Hence~\eqref{eq:trial loc} is an admissible trial state for computing an upper bound on $E (N)$ and the first inequality holds. There remains to evaluate the energy of $\Psiloc$.

\medskip

\noindent\textbf{Step 2: main terms.} To compute the energy we recall 
that~\eqref{eq:ener DM}-\eqref{eq-energy_with_density_matrices} hold with a non-symmetrized 
state $\Psi \in \gH^{\otimes N} $
provided we take as definitions (compare with~\eqref{eq:def DM})
$$ \gammaP ^{(1)} = \sum_{j=1} ^N \tr_{\overline{j}} \left[ |\Psi \rangle \langle \Psi| \right] \quad , \quad 
\gammaP ^{(2)} = \sum_{1 \leq j < k \leq N} \tr_{\overline{\{j,k\} }} \left[ |\Psi \rangle \langle \Psi| \right] \;,
$$
where $\tr_{\overline{j}}$ (respectively $\tr_{\overline{\{j,k\}}}$) stands for the partial trace with respect to all particles but the $j$-th (respectively all particles but the $j$-th and the $k$-th). Another advantage of the unsymmetrized trial 
state~\eqref{eq:trial loc}, apart from the fact that it is normalized when $\Psi_\pm'$ are normalized, 
is that one can easily 
compute its one- and two-body density matrices using these definitions: we easily find 
$$ 
\gamma_{\Psiloc} ^{(1)}= \gamma_{\Psi_-'} ^{(1)} + \gamma_{\Psi_+'} ^{(1)}
\quad , \quad 
\gamma_{\Psiloc} ^{(2)}
= \gamma_{\Psi_- '} ^{(2)}  + \gamma_{\Psi_+ '} ^{(2)}  + \gamma_{\Psi_-'} ^{(1)}\otimes \gamma_{\Psi_+'} ^{(1)} \;.
$$
We insert this in~\eqref{eq:ener DM} and  use $V_N \leq V_N ^\pm$  to obtain
\begin{align}\label{eq:estim ener loc}
\langle \Psiloc | H_N | \Psiloc \rangle & \leq 
\tr \left[(-\Delta + V_N ^- )\gamma_{\Psi_-'} ^{(1)} \right] +
\frac{\lambda}{2(N-1)} \tr \left[w \gamma_{\Psi_-'} ^{(2)} \right]
\nonumber \\
& + \tr \left[(-\Delta + V_N ^+ )\gamma_{\Psi_+'} ^{(1)} \right] 
+ \frac{\lambda}{2(N-1)} \tr \left[w \gamma_{\Psi_+'} ^{(2)} \right]
\nonumber\\
& + \frac{\lambda}{2(N-1)} 
  \tr\left[ w \gamma_{\Psi_- '} ^{(1)}\otimes \gamma_{\Psi_+ '} ^{(1)} \right] \;.
\end{align}
The first two lines  are identical and are estimated as follows
\begin{align*}
\tr \left[(-\Delta + V_N ^\pm )\gamma_{\Psi_\pm '} ^{(1)} \right] + \frac{\lambda}{2(N-1)} \tr \left[w \gamma_{\Psi_\pm '} ^{(2)} \right]
&= \tr \left[(-\Delta + V )\gamma_{\Psi '} ^{(1)} \right] + \frac{\lambda '}{2(N'-1)} \tr \left[w \gamma_{\Psi '} ^{(2)} \right]\\
&= \langle \Psi '| H_{N'} | \Psi ' \rangle \\
& \leq  N' \eH \left( \lambda' \right) + \eB \left( \lambda' \right) + o(1)\;,
\end{align*}
where
the last inequality follows from Proposition~\ref{pro:GSE single well}. 
The last expression gives the desired upper bound in~\eqref{eq:upper bound}, because 
$$\lambda'= \frac{\Delta_N \lambda}{2} = \frac{\lambda}{2} + O (N^{-1})$$
and the discrepancy between $\eB ( \lambda' )$ and $\eB ( \lambda/2 )$
can be easily included in the $o(1)$ term (it is in fact of order $N ^{-1}$, as follows from considerations similar 
to those discussed in Appendix~\ref{sec:app}).
Hence there only remains to estimate the error term on the third line of~\eqref{eq:estim ener loc}, which describes
interactions between the particles in the left well with those in the right well.

\medskip

\noindent\textbf{Step 3: bound on the interactions between particles in different wells.} Let 
$$P_\pm' = | \uHpm^{\lambda'} \rangle \langle \uHpm^{\lambda'}|\: \mbox{ and } \: Q_\pm^{'} = 1 - P_\pm'$$ 
be the orthogonal projectors onto the span of $\uHpm^{\lambda'} = \uH^{\lambda'} ( \cdot \mp \mathbf{x}_N)$ 
and its orthogonal, respectively. 
It follows from the definition of $\Psi_- '$ that
$$ 
P_{-}' \gamma_{\Psi_- '} ^{(1)} P_{-}' = \left( \frac{N}{2} - \tr \left[\gamma^{(1)}_{\PhiBN_-}\right] \right) 
P_{-}'
\quad , \quad 
Q_{-}' \gamma_{\Psi_- '} ^{(1)} Q_{-} ' = \gamma^{(1)}_{\PhiBN_-}\;.
$$ 
Note that, strictly speaking, we only get in the one- and two-body density matrices of 
$\Psi_-$ the contribution from $\PhiBN_- $ living on $j$-particle sectors with $j\leq M$. Using Wick's theorem, one can 
easily see that $\Vert \phiBN _j \Vert^2$ decays very rapidly with $j$, and the contribution from the rest thus yields
 a very small remainder, that we ignore (see similar considerations in Equation~\eqref{eq:decay N bog} below). It is in 
fact sufficient at this stage to notice that 
\begin{equation}\label{eq:diag term}
P_{-}' \gamma_{\Psi_- '} ^{(1)} P_{-}' \leq C N P_{-}'
\quad ,\quad Q_{-} ' \gamma_{\Psi_-'} ^{(1)} Q_{-}' \leq C \gamma^{(1)}_{\PhiBN_-}
\end{equation}
for some constant $C>0$.
Since $\gamma_{\Psi_- '} ^{(1)}$ is a positive self-adjoint operator we have that, as operators,
\begin{equation}\label{eq:off diag term}
 P_{-}' \gamma_{\Psi_- '} ^{(1)} Q_{-}' + Q_{-}' \gamma_{\Psi_- '} ^{(1)} P_{-}' 
 \leq 
  P_{-}' \gamma_{\Psi_- '} ^{(1)} P_{-}' 
+ Q_{-} ' \gamma_{\Psi_- ' } ^{(1)} Q_{-} ' \;.
\end{equation}
Similar formulas holds for $\gamma_{\Psi_+ '} ^{(1)}$.
Writing 
$$ \tr\left[ w \, \gamma_{\Psi_- '} ^{(1)}\otimes \gamma_{\Psi_+' } ^{(1)} \right] 
= \tr\left[ w \, (P_{-}' +Q_{-}')\gamma_{\Psi_-'} ^{(1)} (P_{-}'+Q_{-}') \otimes (P_{+}' + Q_{+}')\gamma_{\Psi_+'} ^{(1)} 
(P_{+}' +Q_{+}') \right]\;,
$$
expanding, inserting~\eqref{eq:diag term} and~\eqref{eq:off diag term}, we thus get the bound
\begin{align}\label{eq:control off site}
 \tr\left[ w \gamma_{\Psi_-'} ^{(1)}\otimes_s\gamma_{\Psi_+'} ^{(1)} \right]
 &\leq 4 C^2 N^2 \tr\left[ w \: P_{-}' \otimes P_{+}' \right] + 
8 C^2 N \tr\left[ w \:  P_{-}' \otimes \gamma_{\PhiBN_+} ^{(1)}\right] \nonumber\\
&+ 4 C^2 \tr\left[ w \:  \gamma_{\PhiBN_-} ^{(1)} \otimes \gamma_{\PhiBN_+} ^{(1)}\right]\;,
\end{align}
where we use that objects in the $V_N^-$ well are the images  of those in the $V_N^+$ well 
under the mirror symmetry $(x^1,x^2,\cdots,x^N) \mapsto  (-x^1,x^2,\cdots, x^N)$ 
to group some terms. 
Since $w$ is bounded and $\gamma_{\PhiBN_-} ^{(1)}$, $\gamma_{\PhiBN_+} ^{(1)}$
are trace-class, the last term is of order one. Recalling that this must be divided by $N-1$ to get the contribution to the 
energy, this is much smaller than the level of precision we aim at. 

For the other two terms we use Lemma~\ref{lem:control pot}: for any trace-class operator $\gamma$ on $\gH$ and any 
$0< \eta<1$, we have 
\begin{align*}
\tr\left[ w \: P_{-}' \otimes \gamma \right] &= \iint_{\R ^d \times \R ^d} |\uHNm (x)| ^2 w(x-y) \gamma (y,y) \D x \D y\\ 
&\leq C_{\eta} \int_{\R ^d} |\uHNm (y)| ^{2-\eta}  \gamma (y,y) \D y
\end{align*}
where we identify $\gamma$ and its kernel. In particular, for the first term of the right-hand side 
of~\eqref{eq:control off site}, we obtain 
$$ \tr\left[ w \: P_{-}' \otimes P_{+}' \right] \leq C_{\eta} \int_{\R ^d} |\uHNm (x)| ^{2-\eta}  |\uHNp (x)| ^2 \D x\;. 
$$
Then we recall that $\uHpm^{\lambda'} = \uHN (.\mp \mathbf{x}_N)$. Using the decay estimate~\eqref{eq:decay Hartree} we have 
$$ \uHm ^{1-\eta} (x) \uHp (x) \leq C_\eps^2 
\exp \Big( -2 (1 - \eps) \big( A (|x-\mathbf{x}_N|) + A ( |x+\mathbf{x}_N| )\big) \Big)\;.
$$
By the same argument as in the proof of \eqref{eq:tunn 1} and by using~\eqref{eq:loc regime},  
we conclude that 
$$N^2 \tr\left[ w \: P_{-}' \otimes P_{+}' \right] \leq
C_\eps' N^2 e^{-4 ( 1 - \eps ) A ( L_N/2)}  \ll 1  
$$
as desired.
Finally, for the second term in the right-hand side of~\eqref{eq:control off site} we write  
\begin{align*}
\tr\left[ w P_{-}' \otimes \gamma_{\PhiBN_+} ^{(1)}\right] 
&\leq C_{\eta} \int_{\R ^d } |\uHNm (x)| ^{2-\eta} \rho_{\PhiBN_+}(x) \D x 
\\
&\leq C_{\eta} \int_{\R ^d} \frac{|\uHNm| ^{2-\eta}}{1 + V_N  ^+ } \: (1+ V_N ^+) \rho_{\PhiBN_+} 
\\
&\leq C_{\eta} \sup_{\R ^d} \frac{|\uHNm| ^{2-\eta}}{1 + V_N  ^+ } \int_{\R ^d} (1+ V_N ^+) \rho_{\PhiBN_+}\;.
\end{align*}
Then, using the decay estimate~\eqref{eq:decay Hartree} again and the fact that $V(r)\to \infty$ as $r \to \infty$, 
 we easily see that 
$$  \sup_{\R ^d} \frac{|\uHNm| ^{2-\eta}}{1 + V_N  ^+ } \to 0 \mbox{ when } N\to \infty$$
whereas Lemma~\ref{lem:decay bog} ensures that 
$$ \int_{\R ^d} (1+ V_N ^+) \rho_{\PhiBN_+} \leq C$$
uniformly in $N$. We thus have 
$$N  \tr\left[ w \:  P_{-}' \otimes \gamma_{\PhiB _+} ^{(1)}\right] \ll N$$
in the limit $N\to \infty,$ which concludes the proof since this term gets divided by $N-1$ in the energy 
expansion \eqref{eq:estim ener loc}.
\end{proof}

\section{Energy lower bound and localization estimate}\label{sec:low bound loc}

In this section, we prove the lower bound corresponding to~\eqref{eq:estim ener loc} by a suitable localization procedure. 
The fluctuations of the number of particles in each well will be estimated in the course of the proof. We 
first split in Sec.~\ref{sec:geom loc} the many-body Hamiltonian into two parts corresponding to the left and right wells. 
For the state of the system we follow the procedure of localization in Fock space presented in~\cite{Lewin-11} 
(see also~\cite[Section 5]{Rougerie-spartacus}) to obtain a lower bound in Sec.~\ref{sec:MF loc} in terms of all the possible ways of 
distributing the particles in the two wells.

\subsection{Geometric localization procedure}\label{sec:geom loc}

Let us first introduce two smooth localization functions, $\chim$ and $\chip$, such that 
$$ \chip ^2 + \chim ^2 = 1$$
and 
\begin{align*}
\supp (\chim) &\subset \left\{ x\in \R ^d \:|\: x^1 \leq \ell \right\} \\
\supp (\chip) &\subset \left\{ x\in \R ^d \:|\: x^1 \geq -\ell \right\} \;,
\end{align*}
where $\ell$ is a localization length satisfying $1\ll \ell \ll L_N$. Clearly, one can assume that
$$
\chip(x^1,x^2,\ldots, x^d)=\chim ( - x^1,x^2,\ldots,x^d) \quad , \quad  
|\nabla \chim | +|\nabla \chip |\leq C\ell ^{-1} \,\1_{ \{| x^1 |\leq \ell\} }
$$
for any $(x^1,\ldots,x^d) \in \R^d$. Next, we define some cut-off functions $\etapm$ along  the $x^1$-direction 
satisfying
$$\etam (x) = 
 \begin{cases}
    0 \;\mbox{ for } x ^1 \leq 0\\
    1 \;\mbox{ for } 0 \leq x ^1 \leq \ell\\
    0 \;\mbox{ for } x^1 \geq 2 \ell.
 \end{cases}
\quad , \quad 
\etap (x) =
 \begin{cases} 
    0 \;\mbox{ for } x ^1 \geq 0 \\
    1 \;\mbox{ for } -\ell \leq x^1 \leq 0  \\
    0 \;\mbox{ for } x^1 \leq - 2 \ell
 \end{cases}
$$
and we consider the two modified potentials
\begin{align}\label{eq:mod pots}
\tilde{V} ^{+}_N &= V_N^{+}  + \left( V_N^{-} - V_N^{+} \right) \etap - |\nabla \chim| ^2 - |\nabla \chip| ^2\nonumber\\
\tilde{V} ^{-}_N &= V_N^{-} + \left( V_N^{+} - V_N^{-} \right) \etam - |\nabla \chim| ^2 - |\nabla \chip| ^2\;.  
\end{align}
Modulo a small perturbation in the strip $\{ x\in \R^d | -2\ell \leq x^1 \leq 2 \ell\}$, 
these two potentials mimic the left and right potentials $V_N^{\pm}$. More precisely, 
\begin{equation}\label{eq:dif pot pert}
0 \leq \delta_N^{\pm} := \frac{V_N ^{\pm} - \tilde{V}_N^\pm}{V_N^{\pm}}  \leq \| \delta_N^\pm \|_\infty \1_{ -2 \ell \leq x ^1 \leq 2 \ell } 
\end{equation}
and it is easy to show that $\| \delta_N^\pm \|_\infty \to 0$ as $L_N \to \infty$ when $V(x) = |x|^s$. 
We have the simple lemma

\begin{lemma}[\textbf{Localizing the Hamiltonian}]\label{lem:loc hamil}\mbox{}\\
Let $\HT^\pm $ be the one-body Hamiltonian with the modified potential  \eqref{eq:mod pots},
$$ 
\HT^\pm := -\Delta + \tilde{V} ^{\pm}_N\;.
$$
For any $\Psi \in \gH ^N$, one has
\begin{equation}\label{eq:split tot ener}
\langle \Psi | H_N | \Psi \rangle 
\geq \sum_{s=\pm} 
\left(
\tr_{\gH} \left[ \HT^{s} \: \chi_{s} \gammaP ^{(1)} \chi_{s} \right] 
+ \frac{\lambda}{2(N-1)}\tr_{\gH ^2} \left[ w \: \chi_{s} ^{\otimes 2} \gammaP ^{(2)} \chi_{s}^{\otimes 2}\right] 
\right)\;.
\end{equation}
\end{lemma}

\begin{proof}
 We split the one-body Hamiltonian using the IMS formula~\cite[Theorem~3.2]{CycFroKirSim-87}
$$ - \Delta = \chim (-\Delta) \chim + \chip (-\Delta) \chip - |\nabla \chim| ^2 - |\nabla \chip| ^2\;.$$
Using also  $\chip^2+\chim^2=1$ and 
$$
V_N (x) = 
V_N^{\pm}(x) + \left( V_N^{\mp} - V_N^{\pm} \right) (x) \eta_{\pm} (x) \;\text{ if }\; x \in \supp ( \chipm)\;,
$$
this yields
\begin{align}\label{eq:split one body}
-\Delta + V_N 
& = \chim \HT^{-} \chim + \chip \HT^{+} \chip\;.
\end{align}
As for the two-body part we note that since $w\geq 0$ we have, for all $\Psi \in \gH ^2$,
\begin{align*}
\langle \Psi, w \Psi\rangle 
&= \iint \left( \chip ^2 (x) + \chim ^2 ( x)\right) \left( \chip ^2 (y) + \chim ^2 (y) \right) w(x-y) |\Psi (x,y)| ^2 \D x \D y\\
&\geq
\iint  \chim ^2 (x) \chim ^2 (y)  w(x-y) |\Psi (x,y)| ^2 \D x \D y +
\iint  \chip ^2 (x) \chip ^2 (y)  w(x-y) |\Psi (x,y)| ^2 \D x \D y 
\end{align*}
and thus, as an operator on the two-body space, 
\begin{equation}\label{eq:split two body}
w \geq \chim ^{\otimes 2} w \chim ^{\otimes 2} + \chip ^{\otimes 2} w \chip ^{\otimes 2}.
\end{equation}
Inserting this into the expressions~\eqref{eq-energy_with_density_matrices}  of the energies
and using the cyclicity of the trace, we get~\eqref{eq:split tot ener}.
\end{proof}

Now we want to see the localized density matrices $\chipm ^{\otimes 2} \gammaP ^{(2)} \chipm^{\otimes 2}$ as the reduced density 
matrices of two states living on the Fock spaces $\cF (\chipm  \gH)$. This is a well-known procedure, recalled 
in~\cite[Section 3]{Lewin-11} and~\cite[Chapter~5]{Rougerie-LMU,Rougerie-spartacus} It is used repeatedly in~\cite{LewNamRou-14,LewNamRou-14c}. To any $N$-body state 
$\Gamma = |\Psi \rangle \langle \Psi |$ (this applies to mixed states also) we associate some localized states 
$G ^-$ and $G ^+ $ in the Fock space 
$\cF(\gH)=\C\oplus\gH\oplus\gH^2\oplus\cdots$,
of the form
\begin{equation}
G ^{\pm} = G_{0}^ {\pm} \oplus G_{1}^ {\pm} \oplus\cdots\oplus G_{N}^ {\pm} \oplus0\oplus\cdots \;,
\label{eq:def_localization}
\end{equation}
with the crucial property that their reduced density matrices satisfy (here we use the convention 
$\tr_{n+1\to n} [ G_n^\pm] = G_n^\pm$)
\begin{equation}
\chipm ^{\otimes n} \gammaP^{(n)} \chipm ^{\otimes n} = \left(G ^{\pm}\right)  ^{(n)} := \sum_{k=n}^N \frac{k!}{(k-n)!}\tr_{n+1\to k}\left[G^{\pm}_{k}\right]\;,
\label{eq:localized-DM} 
\end{equation}
where for any $1 \leq n \leq N$, $\gammaP^{(n)}$ is the $n$-body reduced density matrix of $\Psi \in \gH ^N$ normalized as 
in~\eqref{eq:def DM} and $\gammaP^{(0)}=\chi_{\pm}^{\otimes 0}= 1$.

The relations \eqref{eq:localized-DM} determine the localized states $G^\pm$ uniquely and they ensure that $G ^-$ and $G ^+$ 
are (mixed) states on the Fock spaces $\cF (\chim  \gH)$ and $\cF (\chip \gH)$, respectively:
\bq\label{eq:nomalization-localized-state}
\sum_{k=0}^N \tr \left[ G_{k}^-\right] = \sum_{k=0}^N \tr \left[ G_{k}^+\right]=1.
\eq
An important property is that 
\begin{equation}\label{eq:geom relation}
\tr_{\gH^k} [G^-_{k}] = \tr_{\gH^{N-k}} [G^{+}_{N-k}] \;\mbox{ for all } \; k=0,\ldots,N,
\end{equation}
that is, the probability of having $k$ particles $\chim$-localized is equal to the probability of having $N-k$ particles 
$\chip$-localized.

\medskip

Let us now anticipate a little bit on the forthcoming energy lower bounds. Using the previous constructions, they will be expressed 
in terms of all the possible ways of distributing $n$ particles in one well and $N-n$ particles in the other well. The energy of such 
a configuration will be bounded from below by applying the expansion of Proposition~\ref{pro:GSE single well}, leading to an approximate 
value in terms of $R(n)$ and $R(N-n)$ where
\begin{equation}\label{eq:energy n}
R (n) := n \eH  \left(\lambda\frac{n-1}{N-1} \right) + \eB \left(\lambda \frac{n-1}{N-1}\right) 
\end{equation}
is (to subleading order) the energy of $n$ particles in one well. The key estimate allowing to conclude the proof is contained in Proposition~\ref{lem:distri part}, see Appendix~\ref{sec:app 2},
which confirms that it is more favorable to distribute the particles evenly between the two wells.

\subsection{Lower bound and corollaries}\label{sec:MF loc}

We now complete the proof of the energy estimate in~Theorem~\ref{thm:main loc}:

\medskip

\noindent\textbf{Step 1: splitting the energy.}  Let us define the $n$-body Hamiltonians
\begin{eqnarray}
\nonumber
\tilde{H}_n ^{\pm} & := & 
\sum_{j=1} ^n \left(-\Delta_j + \tilde{V} ^{\pm}_N (x_j) \right)+ \frac{\lambda }{N-1} \sum_{1\leq i< j \leq n} w (x_i-x_j)
\\
\label{eq:hamil pm bis}
H_n ^{\pm} & := & 
\sum_{j=1} ^n \left(-\Delta_j + {V} ^{\pm}_N (x_j) \right)+ \frac{\lambda }{N-1} \sum_{1\leq i< j \leq n} w (x_i-x_j)\;.
\end{eqnarray}
Combining~\eqref{eq:split tot ener} and~\eqref{eq:localized-DM} we obtain the lower bound 
\begin{align}\label{eq:split low bound}
\langle \Psi | H_N | \Psi \rangle \geq \sum_{n=1} ^N \tr_{\gH ^n} \left[ \tilde{H}_n ^{+} G_n ^{+} +  \tilde{H}_n ^{-} G_n ^{-}\right]. 
\end{align}
The rationale in the following is to apply a mean-field approximation
 in each term of the sum in the right-hand side of~\eqref{eq:split low bound} and 
to approximate the 
Hartree energies for the perturbed Hamiltonians $\tilde{H}_n^\pm$ 
by those of the unperturbed ones ${H}_n^\pm$, relying on the considerations of Section~\ref{sec:pert well}.

\medskip

\noindent \textbf{Step 2: mean-field approximation and a-priori bound.} 
We first perform a mean-field approximation in each term of the sum in the right-hand side of~\eqref{eq:split low bound}. 
We regard the operators $\tilde{H}_n ^{\pm}$ as $n$-body Hamiltonians in mean-field scaling with 
effective $n$-dependent coupling constant 
$$
\lambda_n = \lambda \frac{n-1}{N-1}\;.
$$ 
Let $\uHtpm^{\lambda_n}$ be the (unique) minimizer with unit $L^2$-norm
 of the energy functionals $\EHtpm^{\lambda_n} [u]$ obtained by replacing $V$ by 
the perturbed potentials $\tilde{V}^{\pm}_N$ and $\lambda$ by $\lambda_n $
in \eqref{eq:Hartree func}, and  let $\eHtt (\lambda_n)$ be the corresponding Hartree energy
(recall that these energies
are the same for the left and right potentials wells because
$\tilde{V}^{-}_N$ and $\tilde{V}^{+}_N$ are related to each other by the mirror symmetry
$(x^1,x^\bot) \mapsto (-x^1,x^\bot)$).
For $n=1$, $\lambda_n=0$ and  $\uHtpm^{0}$ and $\eHtt (0)= \inf \{ {\rm spec} ( \tilde{h}_\pm)\}$ 
denote respectively the 
ground state and lowest eigenvalue of the one-body Hamiltonian $\tilde{h}_\pm = -\Delta + \tilde{V}^\pm_N$.
Applying Proposition~\ref{pro:bog low} to the functionals  $\EHtpm^{\lambda_n}  [u]$ and recalling~\eqref{eq:hamil fock}, we may 
bound from below each term in the sum of \eqref{eq:split low bound} by
$$
\big(  n \eHtt  \left(\lambda_n \right) - C^2 \big) \big( \tr [G_n ^+] + \tr [G_n ^-] \big) 
 + C \left\langle \dGamma (\tilde{h}_+^\bot) \right\rangle_{G_n ^+} + C \left\langle \dGamma (\tilde{h}_-^\bot) \right\rangle_{G_n ^-}
\;,
$$
where 
$$
\left\langle \dGamma (\tilde{h}_\pm^\bot ) \right\rangle_{G_n ^\pm} 
:= \tr_{\gH^n} \left[ \dGamma (\tilde{h}_\pm^\bot ) \, G_n ^\pm  \right]\;,
$$
$\dGamma  (\tilde{h}_\pm^\bot )$ being the second quantized operator corresponding to the one-body Hamiltonian 
$$ \tilde{h}_\pm^\bot := \tilde{P}^\perp_\pm \left(-\Delta + \tilde{V}^{\pm}_N \right) \tilde{P}^\perp_\pm$$
with $\tilde{P}^\perp_\pm$ the orthogonal projector onto $\{ \uHtpm^{\lambda_n} \} ^{\perp}$.
In view of \eqref{eq:dif pot pert}, one can apply Proposition~\ref{pro:comp Hartree min} 
to the perturbed potentials $\tilde{V}^\pm_N$ to conclude that
$\eHtt ( \lambda_n)$  
is very close to the Hartree energies $\eH (\lambda_n)$ for the unperturbed potentials $V_N^\pm$, with errors of the order
of $|T_N|^{1-\eta}$.
Hence, using also~\eqref{eq:nomalization-localized-state} and~\eqref{eq:geom relation}, we obtain from \eqref{eq:split low bound}
\begin{align} \label{eq-lower_bound_step2}
\nonumber
\langle \Psi | H_N | \Psi \rangle &\geq 
\sum_{n=0} ^N \Big( n \eH  \left(\lambda_n \right) + (N-n) \eH \left(\lambda_{N-n} \right) \Big) \tr [G_n ^+] 
 \\ 
&
+ C \sum_{n=1} ^N 
\left( \left\langle \dGamma (\tilde{h}_+^\bot) \right\rangle_{G_n ^+} + \left\langle \dGamma (\tilde{h}_-^\bot) 
\right\rangle_{G_n ^-} \right)
- C^2 + O ( N | T_N|^{1-\eta}) 
\end{align}
for any $0< \eta <1$.
By Proposition~\ref{cor:tunneling},
$|T_N|^{1-\eta} \ll C_\eta^{1-\eta} e^{-2(1-2 \eta) A ( L_N/2)}$, so that 
$N | T_N|^{-1-\eta}$ converges to zero in the limit~\eqref{eq:loc regime}.
Thus the term $ O ( N | T_N|^{1-\eta}) $ can be absorbed in the constant $C^2$.

The quantity inside the parenthesis in the first line of \eqref{eq-lower_bound_step2} 
gives for large $n$ and $(N-n)$ the ground state energy 
 when one distributes $n$ particles in the left potential well $V_N^-$ and $(N-n)$ particles in the right 
potential well $V_N^+$, the two wells being infinitely
far apart (so that particles in different wells do not interact). 
It is shown in Appendix~\ref{sec:app 2} that  it is more favorable to distribute the particles evenly between the two wells: 
We have for any $n=0,\ldots,  N$,
$$
n \eH ( \lambda_n) + (N-n) \eH ( \lambda_{N-n} ) \geq  N \eH \left( \lambda_{\frac{N}{2}} \right) 
= N \eH \left( \Delta_N \frac{\lambda}{2} \right)
$$
with $\Delta_N$ given by \eqref{eq:delta N}.
Thus, using~\eqref{eq:nomalization-localized-state} again,
\begin{align} \label{eq-proof_a_priori_bound}
\langle \Psi | H_N | \Psi \rangle &\geq N \eH \left( \Delta_N \frac{\lambda}{2} \right) 
+ C \sum_{n=1} ^N \left( 
\left\langle \dGamma (\tilde{h}_+^\bot) \right\rangle_{G_n ^+} + \left\langle \dGamma (\tilde{h}_-^\bot) \right\rangle_{G_n ^-}
\right)
- C' \;.  
\end{align}
Choosing $\Psi$ to be the ground state of $H_N$ and 
combining  with the energy upper bound of Proposition~\ref{pro:loc up bound}, 
we obtain the leading term of the large $N$ expansion of $E_N$,
$$
E_N =  N \eH \left( \Delta_N \frac{\lambda}{2} \right) + O (1)\;,
$$
together with
the following a priori bound  that  will be used below in the estimate of the next-to-leading order terms:
\begin{equation}\label{eq:fluctu number}
0 \leq \sum_{n=1} ^N \left\langle \dGamma (\tilde{h}_+^\bot) \right\rangle_{G_n ^+} 
+  \sum_{n=1} ^N \left\langle \dGamma (\tilde{h}_-^\bot) \right\rangle_{G_n ^-} \leq C''\;.
\end{equation}

\medskip

\noindent \textbf{Step 3: error made by removing the tildes in the lower bound \eqref{eq:split low bound}.}
We now use the a priori bound \eqref{eq:fluctu number} to show that one can replace $\tilde{H}_n^\pm$ by $H_n^\pm$  
in \eqref{eq:split low bound}, making a small error. We first notice that according to \eqref{eq:localized-DM} and
\eqref{eq:hamil pm bis},
$$
\sum_{n=1}^N \tr \left[ ( H_n^+ - \tilde{H}^+_n ) G_n^+ \right]
 =  \tr_{\gH} \left[ \delta_N^+ V_N^+  (G^+)^{(1)} \right]\;,
$$
where $\delta_N^+$ is defined in \eqref{eq:dif pot pert}.
Projecting onto the subspace generated by $\uHtp$  
and its orthogonal, the last trace can be expressed  as a sum of three terms,
$$
\tr_{\gH} \left[ \tilde{P}_+ \delta_N^+ V_N^+ \tilde{P}_+ (G^+)^{(1)} \right]
+ 2 \re \tr_{\gH} \left[ \tilde{P}_+ \delta_N^+ V_N^+ \tilde{P}_+^\bot (G^+)^{(1)} \right]
+  \tr_{\gH} \left[ \tilde{P}_+^\bot \delta_N^+ V_N^+ \tilde{P}_+^\bot (G^+)^{(1)} \right]
\;.
$$
Since 
$$0 \leq (G^+)^{(1)}= \chip \gamma_\Psi^{(1)} \chip \leq  \tr \left[ \chip \gamma_\Psi^{(1)} \chip \right] $$
and $\tr \gamma_\Psi^{(1)} =  N$, the first term is bounded for any $0<\eta <1$ by
$$
0 \leq \tr_{\gH} \left[ \tilde{P}_+ \delta_N^+ V_N^+ \tilde{P}_+ (G^+)^{(1)} \right] 
\leq N  \int_{\{ |x^1| \leq 2 \ell\}} \big| \uHtp |^2 \delta_N^+ V_N^+
 = O ( N |T_N|^{1-\eta})
$$
by virtue of Lemma~\ref{lem:pert enr}. Thus this term converges to zero
in the limit~\eqref{eq:loc regime}.
One deals with the third term by using the identity 
$$
0 \leq \delta_N^+ V_N^+ = \tilde{\delta}_N^+ \tilde{V}_N^+ \; \mbox{ with } \; \tilde{\delta}_N^+ = \frac{\delta_N^+}{1-\delta_N^+}
\;.$$
This gives 
\begin{eqnarray*}
0 \leq \tr_{\gH} \left[ \tilde{P}_+^\bot \delta_N^+ V_N^+ \tilde{P}_+^\bot (G^+)^{(1)} \right] 
& \leq & 
\| \tilde{\delta}_N^+ \|_\infty 
 \tr_{\gH} \left[ \tilde{P}_+^\bot \tilde{V}_N^+ \tilde{P}_+^\bot (G^+)^{(1)} \right] 
\\
& \leq &
\| \tilde{\delta}_N^+ \|_\infty 
 \tr_{\gH} \left[ \tilde{h}_+^\bot (G^+)^{(1)} \right] 
\\
& = & 
\| \tilde{\delta}_N^+ \|_\infty \sum_{n=1}^N \left\langle \dGamma ( \tilde{h}_+^\bot ) 
\right\rangle_{G_n^+} \;,
\end{eqnarray*}
where we used $-\Delta \geq 0$ in the second inequality.
Since $\| \delta_N^+ \|_\infty \to 0$ this term converges to zero too, 
thanks to the a priori bound \eqref{eq:fluctu number}.
Finally, the second term can be treated similarly because the Cauchy-Schwarz inequality gives
\begin{eqnarray*}
& & \left| \tr_{\gH} \left[ \tilde{P}_+ \delta_N^+ V_N^+ \tilde{P}_+^\bot (G^+)^{(1)} \right] \right|^2
\leq 
\tr_{\gH} \left[ \tilde{P}_+ \delta_N^+ V_N^+ \tilde{P}_+ \big( G^+ \big)^{(1)}  \right] 
 \tr_{\gH} \left[ \tilde{P}_+^\bot \delta_N^+ V_N^+  \tilde{P}_+^\bot \big( G^+ \big)^{(1)} \right] 
 \;.
\end{eqnarray*}
Hence
$$
\sum_{n=1}^N \tr \left[ ( H_n^+ - \tilde{H}^+_n ) G_n^+ \right]
= o (1) \;.
$$
The proof for
$H_n^-$ and $\tilde{H}^-_n$ is the same.

\medskip

\noindent \textbf{Step 4: mean-field approximation for the localized energies.} Since we have shown that 
we can discard  
the discrepancy between the original and perturbed functionals, the energy lower 
bound~\eqref{eq:split low bound} yields:
\begin{align}\label{eq:split low bound bis}
\langle \Psi | H_N | \Psi \rangle \geq \sum_{n=1} ^N \tr_{\gH ^n} \left[ H_n ^{+} G_n ^{+} + H_n ^{-} G_n ^{-}\right] 
+ o(1)\;. 
\end{align}
We can now  apply to each of the $n$-body Hamiltonians $H_n^\pm$ 
the bound~\eqref{eq:bog low} of Proposition~\ref{pro:bog low}, which 
includes the corrections to the Hartree energies given by Bogoliubov's theory.
We denote by
\begin{equation}\label{eq:pert part number}
  \Nmpe := \N - a^* (\uHm^{\lambda_n}) a (\uHm^{\lambda_n}) \quad , \quad \Nppe := \N - a^* (\uHp^{\lambda_n}) a (\uHp^{\lambda_n})
\end{equation}
the operators counting the number of particles orthogonal to $\uHm^{\lambda_n}$ and $\uHp^{\lambda_n}$, respectively, with 
$\N$ the total particle number operator. Thus we get
\begin{align*}
 \langle \Psi | H_N | \Psi \rangle &
\geq \sum_{n=1} ^N \Big( n \eH  \left( \lambda_{n} \right) + \eB \left(\lambda_{n} \right) 
- C n ^{-2/5}\Big) \Big( \tr [G_n ^+] + \tr [G_n ^-] \Big)
 \\&+ \sum_{n=1} ^N \frac{1}{n} \left( \left\langle \left(\Nmpe \right)^2 \right\rangle_{G_n ^-} + \left\langle \left(\Nppe \right) ^2 \right\rangle_{G_n ^+}\right) + o(1) \;,
\end{align*}
where  $\eH (0)= \inf \{ \spec ( -\Delta + V_N) \}$ and $\eB (0)=0$ in the term $n=1$.
We next use as before the relation~\eqref{eq:geom relation} to reduce this to 
\begin{align*}
  \langle \Psi | H_N | \Psi \rangle &\geq \sum_{n=0} ^N \left(  E_{n,N-n}^\localized - C  \left( 
n ^{-2/5} \1_{ \{ n>0\} }  +  (N-n) ^{-2/5} \1_{ \{ n<N\} }\right) \right) \tr [G_n ^+] 
  \\&+ \sum_{n=1} ^N \frac{1}{n} \left( \left\langle \left(\Nmpe \right)^2 \right\rangle_{G_n ^-} + \left\langle \left(\Nppe \right) ^2 \right\rangle_{G_n ^+}\right) + o (1)\;,
\end{align*}
where
\begin{equation} \label{eq-energy_localized_distributions}
E_{n,N-n}^\localized = n \eH  \left( \lambda_n  \right) + \eB \left(\lambda_{n} \right)
+ (N-n) \eH  \left( \lambda_{N-n} \right) + \eB \left(\lambda_{N-n} \right)
\end{equation}
is the ground state energy up to $o(1)$ in the case of infinitely far apart wells with 
$n$ particles in the left well and $(N-n)$ particles in the right well (here we set $\eB (\lambda_n): =0$ for $n=0$). 
As before, the energy is minimized by choosing the same number $n=N/2$  of particles
in each well. More precisely, one has (see Proposition~\ref{lem:distri part}  in Appendix~\ref{sec:app 2}) 
\begin{equation}\label{eq:distri even_bis}
E_{n,N-n}^\localized \geq  E_{\frac{N}{2},\frac{N}{2}}^\localized + \frac{C}{N} \left( n-\frac{N}{2}\right)^2
\;,
\end{equation}
so that by ~\eqref{eq:nomalization-localized-state},
\begin{align*} 
 \nonumber
 \langle \Psi | H_N | \Psi \rangle &\geq E_{\frac{N}{2},\frac{N}{2}}^\localized 
+ \frac{C}{N} \sum_{n=0} ^N \left( \left(n-\frac{N}{2}\right) ^2 
-  N n ^{-2/5}  \1_{ \{ n>0\} }  -  N (N-n) ^{-2/5}  \1_{\{ n<N\} }  \right) \tr [G_n ^+] 
  \\&+ \sum_{n=1} ^N \frac{1}{n} \left( \left\langle \left(\Nmpe \right)^2 \right\rangle_{G_n ^-} 
+ \left\langle \left(\Nppe \right) ^2 \right\rangle_{G_n ^+}\right) + o (1)\;.
\end{align*}
Going back to~\eqref{eq-energy_localized_distributions} and using $\eB ( \lambda_{N/2}) = \eB (\lambda/2) +  o(1)$,
we see that the term $ E_{\frac{N}{2},\frac{N}{2}}^\localized$ yields the desired 
first two terms 
in~\eqref{eq:ener lim loc}. To complete the energy lower bound, it thus suffices to
 notice that, for $N$ large enough and any $n= 0, \ldots, N$
$$ 
\left( n-\frac{N}{2}\right) ^2 - N n ^{-2/5}  \1_{ \{ n>0\} }  - N  (N-n) ^{-2/5}  \1_{\{ n<N\} }  \geq \onehalf \left(n-\frac{N}{2}\right) ^2 
- c N ^{3/5}
$$
for some $c>0$, so that 
\begin{align} \label{eq-proof_step4} 
\nonumber
 \langle \Psi | H_N | \Psi \rangle &
\geq N \eH \left( \Delta_N \frac{\lambda}{2} \right) + 2 \eB \left( \frac{\lambda}{2} \right)  + o (1)
  \\
& 
 + \frac{C}{2 N} \sum_{n=1} ^N \left(n-\frac{N}{2}\right) ^2  \tr [G_n ^+] 
+ \sum_{n=1} ^N \frac{1}{n} \left( \left\langle \left(\Nmpe \right)^2 \right\rangle_{G_n ^-} 
+ \left\langle \left(\Nppe \right) ^2 \right\rangle_{G_n ^+}\right)\;.
\end{align}
The energy lower bound 
follows by discarding  the terms on the last line, which are positive. Choosing $\Psi$ in \eqref{eq-proof_step4} to be the ground state of $H_N$ and
combining with the energy upper bound proved in Proposition~\ref{pro:loc up bound}, 
we get as by-products
\begin{equation}\label{eq:fluctu all n}
\sum_{n=1} ^N  \left( n-\frac{N}{2}\right) ^2 \tr [G_n ^+] \ll N 
\end{equation}
and, since $n\leq N$, 
\begin{equation}\label{eq:fluctu N/2}
\sum_{n=1} ^N \left( \left\langle \left(\Nmpe \right)^2 \right\rangle_{G_n ^-} + \left\langle \left(\Nppe \right) ^2 \right\rangle_{G_n ^+}\right) \ll N\;. 
\end{equation}
These estimates provide the control of particle number fluctuations announced in Theorem~\ref{thm:main loc}, as we discuss next.

\subsection{Control of fluctuations}\label{sec:fluctu}

We now conclude the proof of~\eqref{eq:sup fluctu}, using the estimates \eqref{eq:fluctu all n} and 
\eqref{eq:fluctu N/2}. The two terms in the right-hand side of~\eqref{eq:sup fluctu} are estimated similarly, 
let us discuss only one of them. Let us set
\begin{align*}
 \N_- &= a^* \left( \uHm\right) a \left( \uHm\right)\\
 \N_{\chim} &= a^* \left(\chim\uHm\right) a \left(\chim\uHm\right)\;.
\end{align*}
From \eqref{eq:other def DM} and the definition~\eqref{eq:localized-DM}  of 
the localized state $G ^-$, we have
$$
\Big\langle a^\ast ( \chim u )^n a ( \chim v )^n \Big\rangle_{\Psi_N}
= (n!)^{-1} \Big\langle v^{\otimes n} , (G_{-})^{(n)} u^{\otimes n} \Big\rangle 
= \Big\langle a^\ast (u)^n a (v)^n \Big\rangle_{G^{-}}
$$
for any $n=1,\cdots, N$ and $u,v \in \gH$,
so that
$$
\left\langle  \left( \N_{\chim} - \frac{N}{2}\right) ^2 \right\rangle_{\Psi_N}
 = 
\left\langle  \left( \N_- - \frac{N}{2}\right) ^2 \right\rangle_{G ^-} 
+ \Big( \big\| \chim \uHm \big\|_{L^2(\R^d)}^2 - 1 \Big) 
\Big\langle \N_-  \Big\rangle_{G ^-}
\;.
$$
By~\eqref{eq:pert part number}, 
 the operator inequality $(A+B)^2 \leq 2 A^2 + 2 B^2$, and 
$\int \chim^2 |\uHm|^2 \leq 1$, it follows that
$$
\left\langle  \left( \N_{\chim} - \frac{N}{2}\right) ^2 \right\rangle_{\Psi_N}
 \leq  
2 \left\langle  \left(\N - \frac{N}{2}\right) ^2 \right\rangle_{G ^-} 
+ 2 \left\langle \left( \Nmpe \right)^2  \right\rangle_{G ^-}\;.
$$
Recalling the decomposition~\eqref{eq:def_localization} and using ~\eqref{eq:fluctu all n} and~\eqref{eq:fluctu N/2} this gives
\begin{align}\label{eq:fluctu pre final}
\left\langle  \left( \N_{\chim} - \frac{N}{2}\right) ^2 \right\rangle_{\Psi_N}
& \leq  2 \sum_{n=1} ^N  \left( n-\frac{N}{2}\right) ^2 \tr [G_n ^-] 
+ 2 \sum_{n=1} ^N \left\langle \left(\Nmpe \right)^2 \right\rangle_{G_n ^-} \nonumber\\
&\ll  
N \;.
\end{align}
To conclude the proof of~\eqref{eq:sup fluctu}, there only remains to remove the cut-offs function $\chim$. To this end we prove the following simple lemma 

\begin{lemma}[\textbf{Removing cut-offs functions}]\label{lem:remove chi}\mbox{}\\
For any $N$-body bosonic state  $\Gamma_N$, one can find a constant $C>0$ such that  
\begin{equation}\label{eq:dif number}
\left| \left\langle \N_- \right\rangle_{\Gamma_N} - \left\langle \N_{\chim} \right\rangle_{\Gamma_N} \right| \leq  C N \int_{\R ^d} \chip ^2 |\uHm| ^2
\end{equation}
and
\begin{equation}\label{eq:dif number squared}
\left| \left\langle \N_-  ^2 \right\rangle_{\Gamma_N} - \left\langle \N_{\chim} ^2 \right\rangle_{\Gamma_N} \right|  \leq C N^2 \int_{\R ^d} \chip ^2 |\uHm| ^2.
\end{equation}
\end{lemma}

\begin{proof}
We denote 
\begin{align*}
 \OO_1 &= \left| \uHm \right\rangle \left\langle \uHm\right| - \left| \chim \uHm \right\rangle \left\langle \chim \uHm\right|\\ 
 \OO_2 &=  \left( \left| \uHm \right\rangle \left\langle \uHm\right| \right)^{\otimes 2} - 
\left( \left| \chim \uHm \right\rangle \left\langle \chim \uHm\right| \right)^{\otimes 2}.
\end{align*}
Clearly it suffices to prove that 
$$ \left| \tr \left[  \OO_1 \gamma_{\Gamma_N} ^{(1)}\right] \right| \leq C N \int_{\R ^d} \chip ^2 |\uHm| ^2
\quad , \quad  
\left| \tr \left[  \OO_2 \gamma_{\Gamma_N} ^{(2)}\right] \right| \leq C N ^2 \int_{\R ^d} \chip ^2 |\uHm| ^2\;,
$$
where $\gamma_{\Gamma_N} ^{(1)}$ and $ \gamma_{\Gamma_N} ^{(2)}$ are respectively the one- and two-body density matrices 
of $\Gamma_N$, see \eqref{eq:other def DM}. 
But
$$ \left| \tr \left[  \OO_1 \gamma_{\Gamma_N} ^{(1)}\right] \right| \leq \tr\left[ \gamma_{\Gamma_N} ^{(1)}\right] \norm{\OO_1}_{\gS ^\infty}  $$
and 
$$ \left| \tr \left[  \OO_2 \gamma_{\Gamma_N} ^{(2)}\right] \right| \leq \tr \left[ \gamma_{\Gamma_N} ^{(2)} \right] \norm{\OO_2}_{\gS ^\infty} 
\leq 2 \tr \left[ \gamma_{\Gamma_N} ^{(2)} \right]  \norm{\OO_1}_{\gS ^\infty}\;,
$$
where $\gS ^{\infty}$ is the set of compact operators, equipped with the operator norm. Since $\gamma_{\Gamma_N} ^{(1)}$
 and $\gamma_{\Gamma_N} ^{(2)}$ have by definition traces $N$ and $N(N-1)$, it suffices to prove that 
$$ \norm{\OO_1}_{\gS ^\infty} \leq C \int_{\R ^d} \chip ^2 |\uHm| ^2\;.$$
But, as a rank-two operator on $\mathrm{span} \{ \uHm,\chim \uHm\}$, $\OO_1$ has matrix elements
\begin{align*}
\left\langle \uHm, \OO_1 \uHm \right\rangle &= 1 - \left( \int_{\R ^d} \chim | \uHm |^2 \right) ^2\\
\left\langle \uHm, \OO_1 \chim \uHm \right\rangle &= \int_{\R ^d} \chim |\uHm |^2 \left( 1 - \int_{\R ^d} \chim ^2 |\uHm |^2 \right)\\
\left\langle \chim \uHm, \OO_1 \chim \uHm \right\rangle &= \left( \int_{\R ^d} \chim |\uHm |^2 \right) ^2 - \left(\int_{\R ^d} \chim ^2 |\uHm |^2\right) ^2 
\end{align*}
and it is straightforward to see that these are all bounded in absolute value by $C \int_{\R ^d} \chip ^2 |\uHm| ^2$. Hence, so must be the absolute values of the eigenvalues of $\OO$ and we deduce the result.
\end{proof}

The final result~\eqref{eq:sup fluctu} follows from~\eqref{eq:fluctu pre final} and the above lemma, 
recalling that in the regime of our interest we have 
\begin{equation} \label{eq-cross_int}
 \int_{\R ^d} \chip ^2 |\uHm| ^2 \leq \int_{ \{ x^1 \geq \frac{L_N}{2}-\ell \} }  |\uH| ^2 (x) 
= O( |T_N | ) \ll N ^{-1}\;,
\end{equation}
as follows from the choice of the cut-off functions and the decay estimates established in Section~\ref{sec:single well}.

  
\appendix

\section{Fluctuations out of a Bose-Einstein condensate}\label{sec:app}

Let us quickly explain how Proposition~\ref{pro:bog low} follows from the arguments of~\cite{LewNamSerSol-13}. To this end, we let $f$ and $g$ be two smooth truncation functions from $\R^+$ to $\R ^+$, satisfying
$$ f^2 + g ^2 = 1$$
and 
$$f (x) = 0 \mbox{ for } x \geq 1,\quad g (x) = 0 \mbox{ for } x\leq 1/2.$$
Then, define the operators
\begin{equation}\label{eq:trunation}
f_M := f \left( \Npe / M \right), \quad g_M := g \left( \Npe / M \right)  
\end{equation}
on $\cF (\gHp)$, where $\Npe$ is the number operator 
$$ \Npe := \bigoplus_{j=1} ^\infty  j \, \1_{(\gHp) ^j} \;.$$
Let us denote by $\dGamma (h^\perp)$ the second quantization of $h^\perp = P^\perp (-\Delta+V) P^\perp$, acting on $\cF (\gHp)$: 
$$ \dGamma (h^\perp) := \bigoplus_{k=0} ^{\infty} \sum_{j=0} ^k h^\perp_j\;.$$
Recall that $\dGamma (1)$, the second quantization of the identity on $\gHp$, is just $\Npe$.
We argue as follows:

\begin{proof}[Proof of Proposition~\ref{pro:bog low}]
We first pick some $M\leq N $, to be optimized over later, and apply~\cite[Lemma~6.3]{LewNamSerSol-13} to obtain 
$$ \bH_N  \geq f_M \bH_N f_M   +  g_M \bH_N g_M - \frac{C}{M ^2} \left( \dGamma (h^\perp)  + C \right) \;, $$
where we also apply the main results of the same paper to show that the first eigenvalue of $\bH_N$ is bounded by a constant (actually, for large $N$ it converges to the Bogoliubov ground state energy). Next, using~\cite[Proposition~5.1]{LewNamSerSol-13} to estimate the first term, which lives on the smaller space $\cF ^{\leq M} (\gHp)$, we get 
$$ \bH_N \geq \left( 1 -C \sqrt{\frac{M}{N}} \right) f_M \bHb  f_M   +  g_M \bH_N g_M - \frac{C}{M ^2} \left( \dGamma (h^\perp)  + C \right).$$ 
Next, under our assumption that $w\geq 0$ we have
$$ \bH_N \geq \dGamma (h^\perp) \geq C \Npe$$
and (see~\cite[Theorem~2.1]{LewNamSerSol-13})
\begin{equation}\label{eq:op bound bog}
 \bHb \geq C \dGamma (h^\perp +1) - C. 
\end{equation}
Since $f_M$ and $g_M$ commute with $\dGamma(h^\perp +1)$ (the latter conserves the particle number), we may borrow a little part of the main terms to control the error in the above:
\begin{multline*}
 \bH_N \geq \left( 1 -C \sqrt{\frac{M}{N}} - C' M ^{-2} \right) f_M \bHb  f_M   +  (1-C' M ^{-2}) g_M \dGamma (h^\perp) g_M 
 \\ + \frac{ C' - C }{M^2} \dGamma (h^\perp +1) - \frac{C}{M ^2} . 
\end{multline*}
Taking $C'$ large enough to make the first error term positive, and recalling that $h^\perp \geq C > 0$ we arrive at
$$ 
\bH_N \geq \left( 1 -C \sqrt{\frac{M}{N}} -  \frac{C}{M ^2} \right) f_M \bHb f_M   +  C \left( 1- \frac{1}{M ^2}\right) g_M \dGamma (h^\perp) g_M - \frac{C}{M ^2} .  
$$
Next we make the choice $M= N ^{1/5}$ to optimize error terms:
$$
\bH_N \geq \left( 1 -C N ^{-2/5} \right) f_M \bHb f_M   +  C \left( 1- N ^{-2/5} \right) g_M \dGamma (h^\perp) g_M - C N ^{-2/5} . 
$$
To obtain the first inequality in Proposition~\ref{pro:bog low} we may stop at this stage, inserting~\eqref{eq:op bound bog} 
and using the fact that $f_M$ and $g_M$ commute with $\dGamma (h^\perp)$.

We carry on with the proof of~\eqref{eq:bog low}. Since $\eB$ is bounded and $g_M$ localizes on particle numbers larger than
$M /2 \gg \eB$, we clearly have
$$f_M \bHb f_M + g_M \dGamma (h^\perp) g_M\geq f_M \bHb f_M + C g_M \Npe g_M \geq f_M \eB f_M + C \frac{M}{\eB} g_M \eB g_M.$$ 
Then we may write, on $\cF^{\leq N} (\gHp)$,
$$\bH_N \geq \eB + N ^{-2/5} f_M ^2 \Npe   +  C  g_M ^2 \Npe - C N ^{-2/5}. $$ 
Inserting the simple bounds 
$$ f_M ^2 \Npe \leq M f_M ^2, \quad g_M ^2 \Npe \leq N g_M ^2$$
we get 
\begin{align*}
\bH_N &\geq \eB + N ^{-3/5} f_M ^2 \left( \Npe  \right)^2  +  C  N ^{-1} g_M ^2 \left( \Npe  \right)^2 - C N ^{-2/5} \\
&\geq \eB + N ^{-1}  \left( \Npe  \right)^2 - C N ^{-2/5}
\end{align*}
which is the desired final result.
\end{proof}

Now, let us sketch the

\begin{proof}[Proof of Proposition~\ref{pro:GSE single well}]
This is again implicitly contained in~\cite{LewNamSerSol-13}. Using~\eqref{eq:hamil fock} we have 
$$\left\langle \Psi_N | H_N | \Psi_N \right\rangle = c_M \left( N\eH + \left\langle \Phi_B ^M, \bH_N \Phi_B ^M \right\rangle \right)$$
where $\Phi_B ^M$ is the projection of the Bogoliubov ground state onto sectors with less than $M$ particles. Using~\cite[Proposition~5.1]{LewNamSerSol-13} we obtain 
$$ \left\langle \Phi_B ^M, \bH_N \Phi_B ^M \right\rangle = \left\langle \Phi_B ^M, \bHb \Phi_B ^M \right\rangle + O\left( \sqrt{\frac{M}{N}}\right).$$
Applying then~\cite[Lemma 6.2~]{LewNamSerSol-13} we easily get
$$ \left\langle \Phi_B ^M, \bH_N \Phi_B ^M \right\rangle = \left\langle \Phi_B , \bHb \Phi_B  \right\rangle + O (M ^{-2}) + O\left( \sqrt{\frac{M}{N}}\right).$$
With $M\propto N ^{1/5}$ this gives 
$$ \left\langle \Psi_N | H_N | \Psi_N \right\rangle = c_M \left( N\eH + \eB + O (N ^{-2/5})\right)$$
and it remains to estimate $c_M$. Since this constant normalizes $\Psi_N$ in $\gH ^N$ and $\PhiB$ is a state we have 
$$ c_M ^{-2} = \sum_{j=0} ^M \norm{\phiB_j} ^2 = 1 - \sum_{j=M} ^{\infty} \norm{\phiB_j} ^2.$$
But, for any $\delta>0$, 
\begin{align}\label{eq:decay N bog}
 \sum_{j=M} ^{\infty} \norm{\phiB_j} ^2 &\leq \left( \sum_{j=M} ^{\infty} j ^{-\delta} \norm{\phiB_j} ^2 \right) ^{1/2} \left( \sum_{j=M} ^{\infty} j ^{\delta} \norm{\phiB_j} ^2 \right) ^{1/2}  \nonumber\\ 
 &\leq M ^{-\delta /2 }\left\langle \left(\Npe\right) ^{\delta} \right\rangle_{\PhiB} ^{1/2} \leq C_\delta M ^{-\delta /2}
\end{align}
where we use that $\left\langle \left(\Npe\right) ^{\delta} \right\rangle_{\PhiB}$ is finite for any $\delta$. This follows easily from the fact that $\PhiB$ is quasi-free, using Wick's theorem. Hence (again with $M\propto N ^{1/5}$) 
$$ c_M = 1 + O(N ^{-\delta / 2})$$
for any $\delta >0$, which completes the proof.
\end{proof}

Next we turn to the 

\begin{proof}[Proof of Lemma~\ref{lem:decay bog}]
It follows very closely arguments from~\cite[Appendix A]{LewNamSerSol-13} and~\cite{Nam-thesis}. Details are provided for the convenience of the reader. From the expression~\eqref{eq:bog hamil} one can see that the Bogoliubov energy functional can be written as 
\begin{equation}\label{eq:bog func DM}
\EB [\Gamma] := \tr_{\cF_{\perp}} \left[ \bHb \Gamma \right] = \tr \left[ \left( \HH -\muH + \lambda K\right) 
\gamma_\Gamma ^{(1)}\right] + \lambda \Real \tr [K \alpha_\Gamma ]\;, 
\end{equation}
where 
$$\HH = -\Delta + V + \lambda w \ast |\uH| ^2$$
is the mean-field Hamiltonian and $K$ the operator on $\gH= L^2 (\R^d)$ whose kernel is given by
\begin{equation}\label{eq:K}
K(x,y) = \uH (x) w(x-y) \uH (y)\;. 
\end{equation}
Note that 
$$
\langle \psi , K  \psi \rangle 
= \iint_{\R ^d \times \R ^d } \overline{\psi(x)} \uH(x) w(x-y) \uH(y) \psi (y) \D x \D y \\
= \int_{\R ^d} \hat{w} (k) |\widehat{\psi \uH} (k)| ^2 \D k \;,
$$
so it follows from our assumption $\hat{w}\geq 0$ that $K$ is a positive operator. Since $w$ is bounded, $K$ is also trace-class.

The Bogoliubov minimizer $\PhiB$ is the ground state of a quadratic Hamiltonian, in particular it is 
a \emph{pure} quasi-free state. This implies that its one-body and pairing matrices satisfy the relation
\begin{equation}\label{eq:pure QF state}
 \alpha_{\Phi_{\rm B}} \alpha_{\Phi_{\rm B}} ^*  = (\alpha_{\Phi_{\rm B}} J) ^2 = \gamma^{(1)}_{\Phi_{\rm B}} (1+\gamma^{(1)}_{\Phi_{\rm B}})\;, 
\end{equation}
see~\cite[Appendix A]{LewNamSerSol-13}, \cite{Solovej-notes} or~\cite{Nam-thesis}.

We diagonalize the trace-class operator $\gamma ^{(1)}_{\PhiB}$ in the form 
$$ \gamma ^{(1)}_{\PhiB} = \sum_{n\geq 1} c_n |u_n \rangle \langle u_n |
\quad , \quad c_n \geq 0\;,
$$
and the constraint~\eqref{eq:pure QF state} then implies that 
$$ \alpha_{\PhiB} = \sum_{n\geq 1} \sqrt{c_n(1+c_n)} |u_n \rangle \langle \overline{u_n} |$$
with $\ket{\overline{u_n}} := J \ket{u_n}$.
The Bogoliubov energy thus reads
\begin{equation}\label{eq-proof_Lemma4.3}
\eB =  \langle \PhiB, \bHb \PhiB \rangle= \sum_{n\geq 1} 
\Big(
c_n \langle u_n, \left(\HH -\muH\right) u_n\rangle + \lambda c_n \langle u_n, K u_n \rangle 
+ \lambda \sqrt{c_n(1+c_n)} \Real \langle \overline{u_n}, K u_n \rangle
\Big)
\;.
\end{equation}
Since $\left|\langle \overline{u_n}, K u_n \rangle\right| \leq \langle u_n, K u_n \rangle$ and
$c - \sqrt{c(1+c)} > -1/2$ for any $c \in [0,1]$,
we deduce 
$$
\eB \geq  \sum_{n\geq 1} c_n \langle u_n, \left(\HH -\muH\right) u_n\rangle - \frac{\lambda}{2}  \langle u_n, K u_n \rangle  
 = \tr \left[ \left(\HH -\muH\right) \gamma_{\PhiB} ^{(1)} \right] - \frac{\lambda}{2} \tr [K]\;,
$$
where we have used that $K$ is a positive trace-class operator as noted before. 
But $\eB \leq 0$ (see~\eqref{eq:estim bog}), hence
\begin{equation}\label{eq:control bog}
\tr \left[ (\HH- \muH) \gamma_{\PhiB} ^{(1)} \right] \leq \frac{\lambda}{2} \tr [K]\;.
\end{equation}
Recall that $\HH- \muH$ is bounded from below  on $\gHp$ by a positive constant $\kappa>0$
(since $\uH$ is the non-degenerate ground state of $\HH- \muH$, see the proof of Proposition~\ref{pro:comp Hartree min}),
and that $\gamma_{\PhiB} ^{(1)}$ lives on this space. 
Hence we deduce that $\gamma_{\PhiB} ^{(1)}$ is trace-class. 
Furthermore, $\alpha_{\PhiB}$ is Hilbert-Schmidt because of~\eqref{eq:pure QF state}. 
Finally, since both $-\Delta$ and $\lambda w*|\uH| ^2$ are non-negative, we get
$$ \tr \left[ V \gamma_{\PhiB} ^{(1)} \right] \leq 
\tr \left[ \HH \gamma_{\PhiB} ^{(1)} \right] \leq \frac{\lambda}{2} \tr [K] + \muH \tr \left[\gamma_{\PhiB} ^{(1)} \right]\;,
$$
which proves~\eqref{eq-tr_pot_is_bounded}.
\end{proof}

We end this appendix by giving the proof of the lower bound in \eqref{eq:estim bog}. Since $\HH \geq 0$ is bounded from below we obtain from~\eqref{eq:control bog}
$$\tr \left[ \gamma_{\PhiB} ^{(1)} \right]\leq C \lambda.$$
By using \eqref{eq-proof_Lemma4.3}, the inequality $| \langle \overline{u_n} , K U_n \rangle | \leq u_n , K u_n\rangle$
and the
positivity of $K$, we get
\begin{align*}
0 \geq \eB &\geq \tr \left[ \left(\HH -\muH\right) \gamma_{\PhiB} ^{(1)} \right] 
-\lambda \sum_{n\geq 1} \sqrt{c_n(1+c_n)} \langle u_n, K u_n \rangle \\
&\geq \tr \left[ \left(\HH -\muH\right) \gamma_{\PhiB} ^{(1)} \right] -\lambda \left( \norm{K} 
\tr \left[ \gamma_{\PhiB} ^{(1)} \right]\right) ^{1/2} 
\left(\tr [K] + \norm{K} \tr \left[ \gamma_{\PhiB} ^{(1)} \right]\right) ^{1/2}\\
&\geq \tr \left[ \left(\HH -\muH\right) \gamma_{\PhiB} ^{(1)} \right] - C' \lambda ^{3/2} 
\geq - C' \lambda ^{3/2}\;, 
\end{align*}
where the second line follows from  the Cauchy-Schwarz inequality and the last inequality follows from
the fact that $\HH -\muH$ is bounded from below by $\kappa>0$ on $\gH_\perp$.
We may bootstrap the argument to get the claimed lower bound.

\section{Minimal energy when the two wells are infinitely far apart}\label{sec:app 2}

Let us consider the situation in which the distance $L$ between the two potential wells is sent
to infinity before the number of particles $N$. The tunneling energy \eqref{eq:tunneling term} can then be neglected, as well as the
interaction energy  $\int |\uHm |^2 \left(w \ast |\uHp|^2\right)$ between particles in 
different wells. 
The problem can thus be mapped into a problem of two independent interacting bosonic gases localized in
the left and right wells, with fixed particle numbers $n$ and  $N-n$. 
According to Propositions~\ref{pro:bog low} and~\ref{pro:GSE single well}, the corresponding lowest energy
in the large particle number limits $n \gg 1$ and  $N-n \gg 1$ reads
\begin{multline} \label{eq:energy_Fock_states}
E_{n,N-n}^\localized = n \eH  \left(\lambda\frac{n-1}{N-1} \right) + \eB \left(\lambda \frac{n-1}{N-1}\right)  
              \\ + (N-n) \eH  \left(\lambda\frac{N-n-1}{N-1} \right) +  \eB \left(\lambda \frac{N-n-1}{N-1}\right) \;,
\end{multline} 
up to small corrections $o (1)$.
Here,  $\eH(\lambda)$ and $\eB(\lambda)$ are the Hartree and Bogoliubov energies corresponding to the 
Hamiltonian (\ref{eq:hamil depart}) with a single well potential $V_N^+$ or $V_N^{-}$.
Since the number of particles in the left and right wells are equal to $n$ and $N-n$ instead of $N$,  
the coupling constant $\lambda$ must be renormalized as indicated in \eqref{eq:energy_Fock_states}.  

In this appendix, we prove the following very intuitive fact: among all configurations with $n$ particles
in the left well and $N-n$ particles in the right well, 
the configuration with the smallest energy is the one with 
an equal number $n=N/2$ of particles in each well, which has energy
\begin{equation} \label{eq:energy_localized_state}
E_{\frac{N}{2},\frac{N}{2}}^\localized = N \eH \Big( \Delta_N \frac{\lambda}{2} \Big) + 2 \eB \Big( \frac{\lambda}{2} \Big) + o ( N^{-1}) 
\end{equation}
with $\Delta_N$ defined in~\eqref{eq:delta N}. More precisely, we prove the

\begin{proposition}[\textbf{Distributing particles evenly is optimal}]\label{lem:distri part}\mbox{}\\
There exist an integer $N_0$ and a constant $C>0$ such that for any $N \geq N_0$ and 
$0 \leq n\leq N$, 
\begin{equation}\label{eq:distri even}
E_{n,N-n}^\localized \geq  E_{\frac{N}{2},\frac{N}{2}}^\localized + \frac{C}{N} \left|n-\frac{N}{2}\right| ^2
\;.
\end{equation}
\end{proposition}

We will use the following well known property of the Hartree energy.

\begin{lemma}[\textbf{Scaling and convexity of the Hartree energy}]\label{lem:conv Hartree}\mbox{}\\
Let $\eH (m,\lambda)$ be the minimum of the Hartree functional $\EH ^\lambda [u]$ given by 
\eqref{eq:Hartree func} under the constraint $\| u\|^2_{L^2} = m$, 
\begin{equation}\label{eq:ener Hartree masse}
\eH (m,\lambda) := \inf \left\{ \EH ^\lambda [u] \,\big|\, \int_{\R ^d} |u| ^2 = m \right\} \;.
\end{equation}
For any $m,\lambda \geq 0$ we have 
\begin{equation}\label{eq:scaling Hartree}
\eH (m,\lambda) =  m \eH (1,m\lambda) := m \eH ( m \lambda).
\end{equation}
Moreover, $\eH (m,\lambda)$ is a strictly convex function of $m$.
\end{lemma}

\begin{proof}
Equation~\eqref{eq:scaling Hartree} follows from a simple scaling argument. To see the convexity of the energy as a function of the mass, we note that $\EH^\lambda [u]$ is clearly  a strictly convex functional of $\rho = |u|^2$
(see e.g.~\cite[Appendix~A]{LieSeiYng-00} for details). 
We denote by $\rho_{{\rm H} ,1} = |u_{{\rm H} ,1}| ^2$ and $\rho_{{\rm H} ,2} = |u_{{\rm H} ,2}| ^2$ 
the minimizing densities at masses $m_1$ and $m_2 \not= m_1$ and abuse notation by setting
$ \EH^\lambda [\rho_{{\rm H}, i}] := \EH^\lambda [u_{{\rm H},i}]$.
We then have for any $0< t < 1$,
\begin{align*}
t \eH (m_1, \lambda) + (1-t) \eH (m_1, \lambda) &=  t \EH [\rho_{{\rm H},1}] + (1-t) \EH [\rho_{{\rm H},2}] 
\\& >   \EH \left[ t \rho_{{\rm H} ,1} + (1-t) \rho_{{\rm H},2} \right]
\\& \geq \eH \left( t m_1 + (1-t)m_2, \lambda \right)\;,
\end{align*}
where the last inequality comes from 
$$ \int_{\R ^d } \big( t \rho_{{\rm H},1} + (1-t) \rho_{{\rm H},2} \big) = t m_1 + (1-t)m_2.$$
\end{proof}

\begin{proof}[Proof of Proposition~\ref{lem:distri part}]
Using~\eqref{eq:scaling Hartree} we get 
\begin{eqnarray*}
\dd_n  \left( n \eH \left( 1,\lambda\frac{n-1}{N-1}\right)\right)
& = & 
\dd_n  \left( n \frac{N-1}{n-1} \eH \left( \frac{n-1}{N-1}, \lambda \right) \right)\\
& = & - \frac{N-1}{(n-1)^2} \eH \left( \frac{n-1}{N-1}, \lambda \right) + \frac{n}{n-1} \frac{\partial \eH}{\partial m}
 \left( \frac{n-1}{N-1}, \lambda \right)  \\
\dd_n ^2 \left( n \eH \left( 1,\lambda\frac{n-1}{N-1}\right)\right)
& = & 2\frac{N-1}{(n-1) ^3} \eH \left( \frac{n-1}{N-1}, \lambda \right) - \frac{2}{(n-1) ^2} \frac{\partial  \eH}{\partial m} 
\left( \frac{n-1}{N-1}, \lambda \right) \\
& &
+ \frac{n}{(n-1)(N-1)} \frac{\partial^2 \eH}{\partial m^2}  \left( \frac{n-1}{N-1}, \lambda \right). 
\end{eqnarray*}
Since  $\frac{\partial^2 \eH}{\partial m^2} (m,\lambda)$ is strictly positive by Lemma~\ref{lem:conv Hartree} and
$\eH (m,\lambda)$ and $\frac{\partial \eH}{\partial m} (m,\lambda)$ 
are bounded functions of $m$ on $[0,1]$, we deduce that there is a constant $C>0$ such that
$$ \dd_n ^2 \left( n \eH \left( 1,\lambda\frac{n-1}{N-1}\right) \right)\geq \frac{C}{2 N} \quad \mbox{ if } \quad 
n\geq c_N 
$$
with $c_N = O ( N^{2/3})$ as $N \rightarrow \infty$. Moreover (note that the first eigenvalue of the Bogoliubov Hamiltonian is always non-degenerate),
$$ 
\dd_n ^2 \eB \left( \lambda \frac{n-1}{N-1}\right) = \frac{\lambda ^2}{(N-1) ^2} \eB '' \left( \lambda \frac{n-1}{N-1}\right)
$$
is clearly bounded uniformly by a $O(N^{-2})$. Thus, for large enough $N$ one has
$$ \dd_n ^2  E_{n,N-n}^\localized 
\geq \frac{C}{N}\quad  \mbox{ if } \quad n\geq c_N \quad \mbox{ and } \quad N-n \geq c_N\; .
$$
The function $E_{n,N-n}^\localized $ being symmetric around $n=N/2$, this implies that it must have a local minimum there.
One infers from a second-order Taylor expansion at $n=N/2$ and the fact that  the lower bound on the second derivative is uniform
that  
$$
E_{n,N-n}^\localized 
 \geq  
E_{N/2,N/2}^\localized + \frac{C}{N} \left|n-\frac{N}{2}\right| ^2 \quad \mbox{ if } \quad 
n\geq c_N \quad \mbox{ and }\quad  N-n \geq c_N .
$$
To see that the bounds also holds for $n<c_N$ or $N-n<c_N$, we note that for such $n$
$$ E_{n,N-n}^\localized  = N \eH ( \lambda ) + O (N^{2/3}) = N \eH\left( \frac{\lambda}{2}\right) + \frac{C N}{4}+ O (N^{2/3})
\geq  E_{N/2,N/2}^\localized + \frac{C}{N} \Big( \frac{N}{2} - n \Bigr)^2\;,
 $$
where we have used \eqref{eq:energy_localized_state} and the fact that
$\eH (1,\lambda)$ is increasing in $\lambda$. This completes the proof.
\end{proof}
\section{Spin squeezed states}\label{app:BH}

In this appendix we define the spin squeezed states and estimate in the large $N$ limit their energy 
for  the two-mode Bose Hubbard Hamiltonian (see Section~\ref{sec:heuristic2})
\begin{equation*}
  H_{\rm BH}  = 
  e_+ \N_+ + e_- \N_{-} + T_N ( a_{-}^\ast a_+ + a_+^\ast a_{-} ) +
  \frac{U_N}{2} \left( a_+^\ast a_+^\ast a_{+} a_{+} + a_{-}^\ast a_{-}^\ast a_{-} a_{-} \right)\;.
 \end{equation*}
Recall that this Hamiltonian acts on the subspace  $\gH_{\rm BH} \subset \gH ^N$ spanned by the Fock states $\ket{n,N-n}$, $n=0,\cdots, N$, and that  $\N_{-}+\N_{+} = N \1$ in this subspace. Omitting terms proportional to the identity, $H_{\rm BH} $ can be rewritten
as
$$
H_{\rm BH} =  ( e_+-e_{-} ) J_z + 2 T_N J_x + U_N J_z^2\;,
$$
where $J_x$ and $J_z$ are the kinetic momentum operators defined by\footnote{
  It is easy to see that these self-adjoint operators
 satisfy
the usual commutation relations of angular momenta. This implies in particular that
$e^{\I \phi_N J_x} J_z e^{-\I \phi_N J_x} = \cos \phi_N J_z - \sin \phi_N J_y$.
}
%
$$
J_x= \frac{1}{2} ( a_{-}^\ast a_+ + a_+^\ast a_{-} )\mbox{ , } 
J_y= \frac{1}{2\I} ( a_{-}^\ast a_+ - a_+^\ast a_{-} ) \mbox{ , }
J_z =  \frac{1}{2} (a_{-}^\ast a_- - a_+^\ast a_{+} )  = \N_{-}-\frac{N}{2} \1 \;.
$$
The total energy of a state $\Psi \in \gH_{\rm BH}$ invariant under the exchange of the two wells
is thus
\begin{equation} \label{eq-total_energy_bis}
  E_\Psi = \bra{\Psi} H_{\rm HB} \ket{\Psi} =
   2T_N \bra{\Psi} J_x \ket{\Psi}
  + U_N \left\langle \left(\Delta J_z^2\right)^2\right\rangle_\Psi\;.
\end{equation}
An arbitrary state $\Psi \in \gH_{\rm BH}$ can be represented geometrically
by a 3-dimensional vector with components $\bra{\Psi} J_i \ket{\Psi}$, $i=1,2,3$, on the Bloch sphere of radius
$N/2$, together with the corresponding fluctuations  (see e.g.~\cite{Ferrini2011}).

For vanishing interactions $U_N=0$, the ground state of $H_{\rm BH}$ 
is the delocalized state $\Psidloc$ given by \eqref{eq:ansatz deloc}. This state is 
a {\emph{spin coherent state}} centered on the intersection of the Bloch sphere with the $x$-axis,
i.e., it is an eigenstate of
$J_x$ with the highest eigenvalue $N/2$ and has fluctuations of the angular momenta in
the perpendicular directions equal to
$\langle ( \Delta J_y )^2 \rangle_{\rm dloc} = \langle ( \Delta J_z )^2 \rangle_{\rm dloc} = N/4$.
Increasing $U_N/|T_N|$ to small non-zero values, it becomes energetically more favorable to decrease
the particle number fluctuations
$\langle (\Delta \N_{-})^2\rangle= \langle (\Delta J_z)^2\rangle $
and thus the interaction energy (second term in the right-hand side of~\eqref{eq-total_energy_bis}),
to the expense
of increasing a little bit the kinetic and potential energies (first term).
One expects that the ground  state of $H_{\rm BH}$ is a {\emph{particle number spin squeezed state}}~\cite{Kitagawa1993}.
By definition, such a state has reduced fluctuations  
 of $J_z$ (\ie, of $\N_{-}$) and enhanced fluctuations of $J_y$ as compared to the coherent state $\Psidloc$, and like the latter it saturates the spin uncertainty inequality, i.e.,
\begin{equation} \label{spin_uncertainty_relation}
\langle (\Delta J_y)^2 \rangle \langle (\Delta J_z)^2 \rangle = \frac{| \langle J_x \rangle |^2}{4} \;.
\end{equation}
In contrast to coherent states, particles in a squeezed state are
correlated.

A spin squeezed state can be obtained by~\cite{Kitagawa1993}
\begin{equation} \label{eq-squeezed_state}
\begin{array}{lcl}
    \ket{\Psi_{\rm sq}}
    & = &
 \displaystyle     e^{- \I \phi_N J_x} e^{-\I \theta_N J_z^2} \ket{\Psidloc} \\
    & = & 
 \displaystyle  2^{-\frac{N}{2}} \sum_{n=0}^N   \sqrt{\frac{N!}{n!(N-n)!}} e^{-\I \theta_N ( n - N/2)^2}  e^{- \I \phi_N J_x} \ket{n,N-n} 
  \;,
  \end{array}
\end{equation}
where we have used the components 
$c_n = 2^{-N/2} \sqrt{N!/(n!(N-n)!)}$ of $\Psidloc$ in the Fock state basis.
The unitary operator $ e^{-\I \theta_N J_z^2}$ in \eqref{eq-squeezed_state} squeezes the angular momentum fluctuations in one direction while increasing them
in the perpendicular direction, and the unitary $e^{- \I \phi_N J_x}$  rotates the state on the Bloch sphere
around the $x$-axis, in such a way that the squeezing direction be along the $z$-axis.
In fact, choosing 
$\theta_N = N^{-\alpha-1/2}$ and $\phi_N$ given by $\tan \phi_N = N^{\alpha-1/2}$ with an exponent $\alpha \in (1/6,1/2)$,
a lengthly calculation gives in the limit $N\gg 1$
\begin{equation} \label{eq-variance_J_y_J_z}
\langle (\Delta J_z) ^2\rangle_{\rm sq} = \langle (\Delta \N_{-} )^2 \rangle_{\rm sq} \approx
\frac{N^{2 \alpha}}{4}\;\ll\;\frac{N}{4}
\quad , \quad
\langle (\Delta J_y) ^2\rangle_{\rm sq} \approx
\frac{N^{2(1- \alpha)}}{4} \;\gg\;\frac{N}{4}
\end{equation}
and
\begin{equation} \label{eq-expectation_J_x}
\langle \Psi_{\rm sq} | J_x | \Psi_{\rm sq} \rangle = \langle \uHmp\,,\, \gamma_{\Psi_{\rm sq}}^{(1)} \uHpm \rangle
\approx  \frac{N}{2}- \frac{N^{1-2\alpha}}{4}\;,
\end{equation}
so that
$\Psi_{\rm sq}$ satisfies the minimal spin uncertainty condition \eqref{spin_uncertainty_relation}
to leading order in $N$.

The energy of the squeezed state \eqref{eq-squeezed_state} is 
$$
\langle \Psi_{\rm sq} | H_{\rm BH} | \Psi_{\rm sq} \rangle \approx
T_N N \left( 1- \frac{1}{2} N^{-2 \alpha} \right)  + U_N \frac{N^{2 \alpha}}{4} 
$$
with error terms of order $(|T_N|+U_N) N^{|1-4 \alpha|} $.
Comparing with the energy of the coherent state,
$$
\langle \Psi_{\rm dloc} | H_{\rm BH} | \Psi_{\rm dloc} \rangle
= T_N N   + U_N \frac{N}{4}\;,
$$
we find that $\Psi_{\rm sq}$
has a lower energy than $\Psidloc$ 
when $|T_N|/U_N < N^{2 \alpha}/2$.
Since the exponent $\alpha$ can be chosen arbitrary close to $1/2$ and $U_N = O (\lambda N^{-1})$, we may
expect a transition between a delocalized
regime where the ground state of $H_{\rm BH}$ is close to   $\Psidloc$
(Rabi regime) to a localized regime where it is close to a spin squeezed state (Josephson regime) occurring for
$|T_N| \sim \lambda$, as reported in Table~\ref{tab1}.

                                                                                                                                                                                                                                                                                                                 According to \eqref{eq-diagonal_elements_one-body_density_matrix} and \eqref{eq-expectation_J_x},
the one-body density matrix of $\Psi_{\rm sq}$ is almost equal to the density matrix~\eqref{eq:DM 1 dloc} of the delocalized state, up to corrections 
of order $N^{1-2\alpha}$ in the off-diagonal elements. Thus one can conjecture that
in the Josephson regime
$\lambda N^{-2} \ll | T_N| \ll \lambda$,  the  one-body density matrix
$\gamma_{\Psi_N}^{(1)}$ of the ground state is close to
$\gamma_{\rm dloc}^{(1)}$ and has only one macroscopic eigenvalue.
This conjecture and the localization properties of the ground state
reported in Table~\ref{tab1} are supported by
numerical simulations (see e.g.~\cite{Gati_Oberthaler_2007}).

Finally, we note that the state $\Psi$ with  
Gaussian components \eqref{eq-Gaussian_state} considered in Section~\ref{sec:heuristic2}
has  properties similar to $\Psi_{\rm sq}$   in the large $N$ limit.
In fact, choosing $\sigma_N = N^{\alpha}$,
simple calculations  show that $\Psi$ and $\Psi_{\rm sq}$ have  to leading order in $N$ the same
variances of $J_z$ and $J_y$ and expectation of $J_x$, given by \eqref{eq-variance_J_y_J_z} and \eqref{eq-expectation_J_x}.



\end{document}